\newif\ifextended
\extendedtrue

\ifextended
\documentclass[sigconf,screen]{acmart}
\renewcommand\footnotetextcopyrightpermission[1]{}
\else
\documentclass[sigconf,screen]{acmart}
\fi

\usepackage[normalem]{ulem}

\usepackage{hyperref}

\usepackage{latexsym,amsmath,amsfonts,stmaryrd,mathtools}
\usepackage{amsthm}
\usepackage{algorithm, algorithmicx}
\usepackage[noend]{algpseudocode}
\usepackage{graphicx}
\usepackage{wrapfig}
\usepackage{listings, fancyvrb}
\usepackage{multirow}
\usepackage{comment}
\usepackage{xspace}
\usepackage{pifont}
\usepackage{parcolumns}
\usepackage{graphicx}
\usepackage{enumerate}
\usepackage{enumitem}
\usepackage{wrapfig}
\usepackage{tabularx}
\usepackage{titlecaps}
\usepackage{booktabs}
\usepackage{upgreek}

\usepackage[scaled=0.9]{beramono} 

\microtypecontext{spacing=nonfrench}

\newtheorem{theorem}{Theorem}[section]
\newtheorem{lemma}[theorem]{Lemma}
\newtheorem{corollary}[theorem]{Corollary}

\newtheorem{observation}{Observation}

\makeatletter
\lst@Key{countblanklines}{true}[t]%
{\lstKV@SetIf{#1}\lst@ifcountblanklines}

\lst@AddToHook{OnEmptyLine}{%
	\lst@ifnumberblanklines\else%
	\lst@ifcountblanklines\else%
	\advance\c@lstnumber-\@ne\relax%
	\fi%
	\fi}
\makeatother

\usepackage{caption} 
\newcommand{\sysname}{uBFT\xspace}

\newcommand{\tcbfull}{Consistent Tail Broadcast\xspace}
\newcommand{\tcb}{CTBcast\xspace}

\newcommand{\cbfull}{Consistent Broadcast\xspace}
\newcommand{\cb}{CBcast\xspace}

\newcommand{\tbfull}{Tail Broadcast\xspace}
\newcommand{\tb}{TBcast\xspace}

\newcommand{\us}{µs\xspace}

\newcommand{\summary}{summary\xspace}
\newcommand{\summaries}{summaries\xspace}

\renewcommand{\paragraph}[1]{\medskip\noindent\textbf{#1}}

\def\t{\textit}

\ifextended
\fi

\AtBeginDocument{%
  \providecommand\BibTeX{{%
    \normalfont B\kern-0.5em{\scshape i\kern-0.25em b}\kern-0.8em\TeX}}}

\ifextended
\acmConference[ASPLOS '23]{the 28th
ACM International Conference on Architectural Support for Programming
Languages and Operating System}{March 25--29, 2023}{Vancouver, Canada [Extended]}
\else
\copyrightyear{2023}
\acmYear{2023}
\setcopyright{acmlicensed}
\acmConference[ASPLOS '23]{Proceedings of the 28th ACM International Conference on Architectural Support for Programming Languages and Operating Systems, Volume 2}{March 25--29, 2023}{Vancouver, BC, Canada}
\acmBooktitle{Proceedings of the 28th ACM International Conference on Architectural Support for Programming Languages and Operating Systems, Volume 2 (ASPLOS '23), March 25--29, 2023, Vancouver, BC, Canada}
\acmPrice{15.00}
\acmDOI{10.1145/3575693.3575732}
\acmISBN{978-1-4503-9916-6/23/03}

\fi

\begin{document}

\ifextended
\title{uBFT: Microsecond-Scale BFT using Disaggregated Memory \textnormal{[Extended Version]}}
\fi


\author{Marcos K. Aguilera}
\email{maguilera@vmware.com}
\affiliation{%
  \institution{VMware Research}
  \country{United States}
}

\author{Naama Ben-David}
\email{bendavidn@vmware.com}
\affiliation{%
  \institution{VMware Research}
  \country{United States}
}

\author{Rachid Guerraoui}
\email{rachid.guerraoui@epfl.ch}
\affiliation{%
  \institution{EPFL}
  \country{Switzerland}
}

\author{Antoine Murat}
\email{antoine.murat@epfl.ch}
\affiliation{%
  \institution{EPFL}
  \country{Switzerland}
}

\author{Athanasios Xygkis}
\email{athanasios.xygkis@epfl.ch}
\affiliation{%
  \institution{EPFL}
  \country{Switzerland}
}

\author{Igor Zablotchi}
\email{igorz@mit.edu}
\affiliation{%
  \institution{MIT}
  \country{United States}
}


\ifextended
\renewcommand{\shortauthors}{Aguilera, et al.}
\fi

\begin{abstract}
We propose uBFT, the first State-Machine Replication (SMR) system to achieve microsecond-scale latency in data centers, while using only $2f{+}1$ replicas to tolerate $f$ Byzantine failures.
The Byzantine Fault Tolerance (BFT) provided by \sysname is essential as pure crashes appear to be a mere illusion with real-life systems reportedly failing in many unexpected ways. %
uBFT relies on a small non-tailored trusted computing base---disaggregated memory---and consumes a practically bounded amount of memory (both local and disaggregated).
uBFT is based on a novel abstraction called \tcbfull, which we use to prevent equivocation while bounding memory. 
We implement uBFT using RDMA-based disaggregated memory and obtain an end-to-end latency of as little as 10\,\us.
This is at least 50$\times$ faster than MinBFT%
, a state-of-the-art $2f{+}1$ BFT SMR system based on Intel's SGX. 
We use \sysname to replicate two key-value stores (Memcached
and Redis), as well as a financial order matching engine (Liquibook).
These applications have low latency (up to 20\,\us) and become Byzantine tolerant with
  as little as 10\,\us more.
The price for \sysname is a small amount of reliable disaggregated
  memory (less than 1\,MiB), which in our prototype
  consists of a 
  small number of memory servers connected through RDMA and replicated
  for fault tolerance.

\end{abstract}



\maketitle

\section{Introduction}

Data center applications such as social networks, web search, e-commerce, and banking 
increasingly need both microsecond-scale performance~\cite{mukiller,mumarket,nanopu,putting-mu-back-in-muservices} and strong fault tolerance,
in order to deliver on their promise of being the backbone of today's online services.
Indeed, strong fault tolerance is required by real-life systems which reportedly fail in many unexpected ways. 
Apart from simple crashes, failures in distributed systems 
range from software/configuration bugs~\cite{taxdc, why-do-internet-services-fail}, 
to hardware failures~\cite{flash-memory-failures, ssd-failures, dc-net-reliability}, 
to hardly detectable hardware bugs~\cite{cores-that-dont-count, dattatraya2022detecting, fail-slow-at-scale, gray-failure} 
and up to malicious activity~\cite{sharing-the-data-center, private-eye}.
Traditionally, protecting against the plethora of different failures required slow and expensive Byzantine Fault Tolerant (BFT) protocols.

The standard way to achieve fault tolerance is state-machine replication (SMR).
Typical BFT SMR protocols incur milliseconds of latency~\cite{bft-smart, hyperledger},
    require a large number of replicas
         ($3f{+}1$ to tolerate $f$ failures)~\cite{castro1999practical, hotstuff, honeybadger},
    consume unbounded memory~\cite{VeroneseCBLV13}, and/or
    rely on a large trusted computing base~\cite{behl2017hybrids, KapitzaBCDKMSS12, VeroneseCBLV13}.
  These reasons might explain why BFT has had no adoption in data centers.

In this paper, we propose \sysname, the first
BFT SMR system
  that simultaneously offers four key features:
  (1)~microsecond-scale latency,
  (2)~few replicas ($2f{+}1$),
  (3)~practically bounded memory, and
  (4)~a small non-tailored trusted computing base.
In the common case,
\sysname leverages unanimity to replicate requests in as little as 10\,\us end-to-end
without invoking the trusted
  computing base or expensive cryptographic primitives.
In the slow path---when there are failures or slowness in the network---\sysname
  uses a novel protocol that combines digital signatures with
  judicious use of a trusted computing base.
The trusted computing base in \sysname is non-tailored and small: rather than
  trusted enclaves with arbitrary logic such as Intel's SGX~\cite{sgx-explained}
  or trusted hypervisors~\cite{zhuo2014mft}---which
  have large attack
  surfaces due to their complexity~\cite{sgx-vulnerabilities-survey1, smashex}---\sysname relies solely on disaggregated
  memory, a technology increasingly present in data centers
  due to the availability
  of RDMA~\cite{rdma-manual} today and CXL~\cite{intelcxl} in a few years.
The key mechanism from disaggregated memory we leverage in \sysname is single-writer
  regions (regions of memory that can be written by one designated host and read by others), which are implemented in hardware through access permissions.

Providing the above four features is challenging for BFT protocols.
To get microsecond-scale latency, BFT protocols need to avoid
  expensive public-key cryptography and reduce communication rounds in the
  common path, and doing so has typically required increasing rather
  than decreasing the number of replicas~\cite{kotla2007zyzzyva, qubft, fast-byzantine-consensus}.
Meanwhile, decreasing the number of replicas has usually required
  unbounded memory, sophisticated, or tailored trusted computing bases such as
  append-only-memory~\cite{a2m}, SGX~\cite{behl2017hybrids}, TrInc~\cite{levin2009trinc},
  or 
  reliable hypervisors~\cite{zhuo2014mft}.
Limiting the amount of memory is a significant challenge in the
  design of \sysname, as the standard technique to handle Byzantine behavior in systems with $2f{+}1$ replicas requires storing all messages received, leading to
  long message histories~\cite{aguileraBGPXZ21, VeroneseCBLV13},
  which consume unbounded memory.
Finally, not tailoring the trusted computing base to
  our needs requires designing around existing technologies---in our case
  disaggregated memory---rather than custom hardware.

To respond to these challenges, \sysname introduces a new abstraction called
  \tcbfull (\tcb) that we use to prevent
  equivocation~\cite{on-the-subject-of-non-equivocation}, while requiring a practically bounded
  amount of memory.\footnote{
    Memory consumption scales logarithmically with the number of operations. With the exception of sequence numbers, all values are represented using a bounded number of bits.
  }
Equivocation---a major source of problems in a system with Byzantine
  failures~\cite{a2m}---occurs
  when a faulty process incorrectly 
  sends different information
  to different processes, which may cause the state of
  replicas to diverge.
\tcb prevents equivocation for all messages, but only ensures the delivery of the last $t$
  broadcast messages, where $t$ is a parameter that trades memory for latency (we explain how to set it in Section \ref{sec:evaluation}).
 
The price for \sysname is a small amount (less than 1\,MiB) of reliable disaggregated memory.
\sysname is designed modularly to work with a generic such component;
  our
  current prototype implements this component
  using RDMA and a set of memory nodes
  that themselves are replicated for fault
  tolerance.
These nodes add to the total number of replicas,
  but these replicas are tiny and simple: they do not
  store the state of the application, just
  a few in-flight coordination messages.
Moreover, their
  functionality is application-oblivious,
  so they can be shared among many 
  applications, amortizing their cost.
The memory nodes that provide the disaggregated memory constitute the trusted computing
  base in our prototype and are assumed to fail only by crashing.
This shrinks the vulnerability of the system compared to currently deployed
  crash-tolerant SMR systems, in which \emph{all} components can fail only
  by crashing, effectively making the trusted computing base be the
  entire data center.

We evaluate \sysname against two state-of-the-art systems.
First, we compare it against Mu,
  the fastest SMR system to our knowledge,
  but that tolerates only crash
  failures.
Compared to Mu, \sysname increases the end-to-end latency by only $2\times$, while tolerating Byzantine failures.
Second, we compare \sysname against MinBFT, a state-of-the-art $2f{+}1$ BFT SMR system, and showcase that our system has more than $50\times$ and $2\times$ better latency when operating in its fast and slow path, respectively.
We also use \sysname to replicate two low-latency key-value stores (Memcached~\cite{memcached} and Redis~\cite{redis}), 
and a financial order matching engine (Liquibook~\cite{liquibook}).
All these applications have request latencies of less than 20\,\us when unreplicated and become Byzantine-resilient with as little as 10\,\us more.

In summary, our main contributions are the following:
\begin{itemize}
    \item The design of \sysname, a BFT system for state-machine replication
          with microsecond-scale latency in the
          common case, using only $2f{+}1$ replicas, practically bounded memory, and a small trusted computing base (disaggregated memory).
    \item A new abstraction against equivocation, \tcbfull (\tcb), and a
          protocol for \tcb that uses a small amount of disaggregated memory and has a signature-less fast path.
    \item An open-source implementation of \sysname, \tcb, and reliable shared disaggregated memory using RDMA, available at \url{https://github.com/LPD-EPFL/ubft}.
    \item A thorough evaluation of the performance of \sysname and its applications.

\end{itemize}

\section{Background} \label{sec:background}

\subsection{State-Machine Replication}
State-Machine Replication (SMR) is a universal technique for replicating a service
(e.g., a database) across multiple machines, called \emph{replicas}, to tolerate failures.
The core idea of SMR is that all replicas receive and execute the same client requests
in the same order, ensuring that their states remain consistent.

At the heart of any SMR system lies consensus~\cite{cachin2011rachid}, which provides
total ordering by ensuring that all replicas agree on the same sequence of client
requests.
Consensus protocols are often leader-based~\cite{castro1999practical, rbft, hunt2010zookeeper, mencius, epaxos, bftbench}:
a replica designated as the \emph{leader} orders client requests and forwards them
to the other \emph{follower} replicas to ensure agreement.
Each replica then executes the ordered requests on its local copy of the service
and forwards the resulting output to the clients.

Despite their close relationship, consensus and SMR differ in their memory requirements.
Consensus must guarantee agreement among participants indefinitely; consequently,
decided requests may need to be retained forever to accommodate unresponsive participants,
leading to unbounded memory usage.
In contrast, SMR does not require storing the entire request history, as it focuses
only on replicating the application state, which may be finite even in the presence
of an unbounded number of requests.
As a result, SMR systems must adapt consensus protocols to operate using finite
memory~\cite{shahmirzadi2009relaxed}.

\subsection{Non-Equivocation}
Byzantine processes can equivocate, i.e., they can maliciously say different things to different processes.
In SMR, specifically, a Byzantine leader may propose different client requests to try to cause replicas to diverge,
justifying why
SMR protocols must ensure non-equivocation.

Under the Byzantine asynchronous model, $3f{+}1$ replicas are needed to prevent equivocation~\cite{lamport1982byzantine}.
However, if equivocation is prevented and transferable authentication is available, Byzantine SMR requires only $2f{+}1$ replicas~\cite{correia2004half}, the same number as in the crash-stop case.
With transferable authentication, a process that verifies a proof about the origin of a message can transfer the proof to other processes and be assured they can also verify it.
For example, digital signatures provide transferable authentication, while arrays of Message Authentication Codes (MACs) do not~\cite{matrix-signatures}.

Preventing equivocation using up to $2f{+}1$ replicas requires a compromise~\cite{correia2004half}, i.e., a hybrid model where part of the system---called the \emph{trusted computing base}---fails only by crashing.
Ideally, this base is as small as possible, since a small and simple base is less likely to be susceptible to Byzantine failures (e.g., vulnerabilities, bugs).

\subsection{Disaggregated Memory} \label{sec:background:disagmem}

Disaggregated memory is an emerging data-center architecture that separates compute and memory resources by providing a shared memory pool that compute nodes access over a network.
The memory pool has limited compute capabilities, which it uses for management tasks such as connection handling.
Disaggregated memory improves memory utilization, separates the scaling of
compute and memory, and achieves better availability due to the separation of fault domains~\cite{dm-open-problems}.

Disaggregated memory can be provided by different technologies.
The emerging CXL standard will support disaggregated memory in the future~\cite{directcxl}, while today disaggregated memory is available via 
Remote Direct Memory Access (RDMA)~\cite{rdma-manual} on InfiniBand~\cite{pfister2001infiniband} or RoCE~\cite{beck2011roce-performance}.
RDMA is a networking technology that allows a process to read or write the memory of another machine without involving the CPU of the latter.
Combined with kernel-bypass, RDMA enables sub-microsecond communication and stringent tail latency.
A key feature of RDMA that we leverage is the ability to set access rights to RDMA-exposed memory individually for each accessor.

\begin{figure}
    \centering
    \includegraphics[width=1.0\columnwidth]{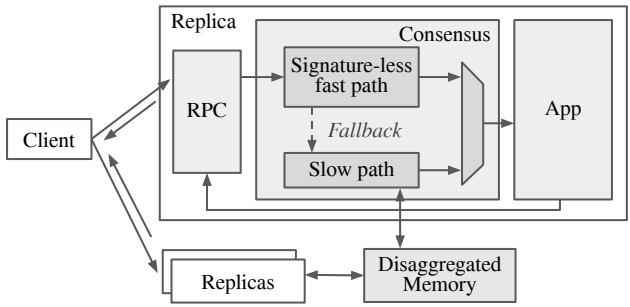}

    \caption{Overview of \sysname's architecture.}
    \label{fig:arch}
\end{figure}

\subsection{Model}
\label{sec:background:model}

We consider a  system with $2f{+}1$ compute nodes and single-writer multiple-reader disaggregated memory. 
Up to $f$ compute nodes are Byzantine and may thus fail arbitrarily.
We assume network connections are authenticated and tamper-proof (processes know who they get messages from and messages cannot be altered) and 
    eventually available (network partitions are intermittent).
We also assume that the disaggregated memory is trusted: it may
  fail only by crashing.
The disaggregated memory is divided into chunks, where each chunk is readable by all compute nodes
  and writable by a designated compute node.
We assume the existence of public-key cryptography:
  processes can sign messages using their private key and verify unforgeable signatures using the pre-published public keys of all processes.
We further assume eventual synchrony: network and processing delays are unbounded until an unknown \emph{Global Stabilization Time} (GST) after which delays are bounded by a known $\delta$.
Lastly, our
system assumes bounded clock drift for safety, i.e., the clocks of correct processes drift from each other
  with a bounded rate.
These assumptions are common for distributed systems in data centers~\cite{castro1999practical, aguilera2009no-time-for-asynchrony, falcon, yang2019synchronous-dc, farm2, ukharon}.

In our prototype,
we do \emph{not} assume that we are given a reliable disaggregated memory~\cite{carbink,hydra}, but rather show how to implement
  a reliable disaggregated memory
  using RDMA.\footnote{Our design encapsulates disaggregated memory, making it possible to replace the current RDMA implementation with a CXL-based alternative in the future.}
To do so, we assume $2f_m{+}1$
  memory nodes out of which $f_m$ can fail.
Memory nodes are part of the trusted computing
  base: they are not Byzantine and may fail by
  crashing only.
Memory nodes are simple: they just provide
  read and write functionality with access
  control.
Their size
  and functionality do not depend on the application
  being replicated, and they can be
  shared among many applications.

\section{Design} \label{sec:design}

\subsection{Overview}

\sysname follows the design of PBFT~\cite{castro1999practical}, a seminal paper that describes how to build practical BFT SMR systems.
Figure~\ref{fig:arch} depicts the architecture of \sysname.
On the left, a client sends requests to the replicas on the right and waits for responses from a majority of them. The replicas process requests in two stages. First, they establish a total order on client requests using a leader-based BFT consensus protocol. Second, they execute the ordered requests on their local instance of the replicated application and then forward the results back to the client.
To achieve microsecond-scale latency, \sysname's consensus engine uses a fast/slow path approach:
as long as the system is synchronous and all replicas collaborate, the fast path 
  orders requests without signatures.
If the fast path does not make progress, \sysname's consensus switches to the slow path, which makes progress with a mere majority of processes using signatures and disaggregated memory.

\sysname significantly 
differs from PBFT in the way it prevents equivocation.
PBFT, being a $3f{+}1$ BFT protocol, relies on intersecting quorums to ensure that malicious replicas do not make the state of honest replicas diverge.
By contrast, \sysname operates with $2f{+}1$ replicas, and therefore cannot rely on the same mechanism as PBFT;
instead, it relies on trusted disaggregated memory, which is encapsulated within a new primitive called \emph{\tcbfull} (\tcb).

\tcb is a variant of Consistent Broadcast~\cite{aguileraBGPXZ21}.
Consistent Broadcast prevents equivocation by ordering all messages broadcast by a given process.
With it, a Byzantine leader is constrained from sending different request orderings to different followers.
Our \emph{tail} variant
is a relaxation that requires correct processes to
deliver 
only the last $t$ messages sent by a correct broadcaster, while
preserving non-equivocation for all messages.
This relaxation is essential to practically bound the memory use.
Importantly, our implementation of \tcb uses a signature-less fast path to meet \sysname's latency requirements.

\begin{figure}
    \centering
    \includegraphics[width=1.0\columnwidth]{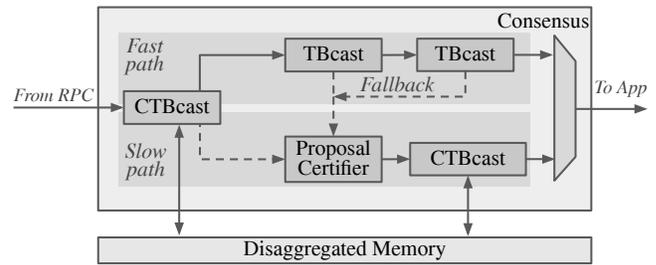}
    \caption{Overview of \sysname's consensus engine.}
    \label{fig:consensus-arch}
\end{figure}

Figure~\ref{fig:consensus-arch} depicts \sysname's consensus component with its fast/slow path design.
After receiving a request from RPC, the leader proposes its ordering via a round of \tcb.
The rest of consensus tries to turn this ordering into a globally accepted one (i.e., stable across leaders).
Depending on the synchrony of the system and the number of faulty replicas, this round of \tcb might execute either in its fast or in its slow path.
In the former case, consensus continues with its fast path and executes two rounds of \tbfull (\tb), a form of best-effort broadcast designed for finite memory
(\S\ref{sec:tcb:def}).
Importantly, none of these three broadcasts involve signatures nor disaggregated memory.
If liveness is lost during the fast path of consensus, \sysname activates its slow path, shown by the dashed arrows, which executes a certification round
  and another round of \tcb.
The slow path is also executed if the fast path of the initial \tcb fails.
Differently from the three rounds of the fast path, the three rounds of the slow path all require signatures, and \tcb invocations must access disaggregated memory.

\subsection{Challenges}

uBFT addresses the following challenges:

\paragraph{2f+1 Replicas and Finite Memory.}
Previous theoretical work~\cite{aguilera2019impact} proposed to prevent equivocation with fewer than $3f{+}1$ replicas by building Consistent Broadcast on top of shared registers.
However, this abstraction requires replicas to use infinite memory in order to store and deliver \emph{all} broadcast messages, which is not implementable in practice.
We work around this memory issue by designing \tcb, a weaker form of \cbfull where replicas are allowed to skip the delivery of old messages in order to favor the delivery of newer ones (\S\ref{sec:tcb}).

\paragraph{SMR with \tcb.}
The reliance of \sysname on \tcb brings additional complexity to its consensus algorithm, notably on preventing equivocation \emph{across} messages.
Typically, protocols rely on the entire history of messages to prevent equivocation.
Yet, \tcb only guarantees the delivery of the tail, which may lead correct replicas to have gaps in their delivery history.
\sysname works around this limitation via \emph{\tcb \summaries}, which allow a replica to make progress in spite of gaps (\S\ref{sec:consensus}).

\paragraph{Microsecond-Scale Operation.}
Systems that operate at the microsecond scale should avoid signatures on their critical path.
Yet, Aguilera \textit{et al.}~\cite{aguileraBGPXZ21} show that Consistent Broadcast cannot completely remove signatures.
Moreover, recycling memory requires the generation of cryptographic proofs which also involve signatures.
\sysname addresses this challenge by avoiding expensive cryptography in the fast path of \tcb and relegating the few bookkeeping signatures to a background task (\S\ref{sec:fast-path}).

\paragraph{Resilient Disaggregated Memory.}
\sysname relies on RDMA to implement disaggregated memory.
However, raw memory exposed over RDMA is not enough to implement our SMR protocol.
Indeed, 
RDMA-exposed memory does not tolerate failures,
and data accesses can be inconsistent, since RDMA provides only 8-byte atomicity.
\sysname addresses these limitations of RDMA using efficient, yet Byzantine-tolerant, algorithms (\S\ref{sec:implementation}).

\section{\tcbfull} \label{sec:tcb}

\emph{\tcbfull} (\tcb) is a novel variant of \cbfull (\cb) that \sysname uses to prevent equivocation.
Briefly, \tcb resembles \cb, except that it allows processes not to deliver outdated messages.
In this way, \tcb avoids maintaining the full history of messages, to bound memory use.

\subsection{Definition} \label{sec:tcb:def}

\tcb is defined in terms of two primitives, \textit{broadcast(k,m)} and \textit{deliver(k,m,p)}, 
where $k$ is a numeric identifier, $m$ is a message, and $p$ is a process.
When $p$ invokes \textit{broadcast(k,m)}, we say that $p$ broadcasts $(k,m)$, i.e., it broadcasts message $m$ with identifier $k$.
A correct broadcaster increments $k$ sequentially at every broadcast, starting with $k{=}1$.
Similarly, when a process $q$ invokes \textit{deliver(k,m,p)}, we say that $q$ delivers $(k,m)$ from $p$.

In simple terms, \tcb is a multi-shot abstraction that prevents correct processes from delivering different messages 
  from a given broadcaster $p$
  for the same identifier $k$.
\tcb is parameterized by a tail $t$, which specifies which messages are guaranteed to be delivered.
Informally, in \tcb, a correct process $q$
is only required to deliver the last $t$ messages
broadcast by a correct process $p$, while the delivery of previous messages is best-effort.

Formally, \tcb has the following properties:
\begin{description}
 \item [Tail-validity] If a correct process $p$ broadcasts $(k, m)$ and never broadcasts a message $(k', m')$ with $k' \ge k + t$, then all correct processes eventually deliver $(k, m)$.
 \item [Agreement] If $p$ and $q$ are correct processes, $p$ delivers $(k, m)$ from $r$, and $q$ delivers $(k, m')$ from $r$, then $m=m'$.
 \item [Integrity] If a correct process delivers $(k, m)$ from $p$ and $p$ is correct, $p$ must have broadcast $(k, m)$. 
 \item [No duplication] No correct process delivers $(k, *)$ from $p$ twice. 
\end{description}

The difference between \tcb and \cb lies in their validity property.
Tail-validity implies that a correct process is only obliged to deliver a message $m$ from $p$ if $m$ is among the last $t$ messages broadcast by $p$.
When $t{=}\infty$, tail-validity reduces to \cb's validity.

The infinite tail of \cb is what prevents it from recycling memory.
Indeed, given that the broadcaster cannot distinguish between network asynchrony and receiver failures~\cite{failure-detectors},
it is required to keep re-transmitting all messages until they are explicitly acknowledged.
Thus, in \cb, the broadcaster can garbage collect messages only after they have been acknowledged by all receivers.
As a result, once a single process fails, memory cannot be recycled and any correct implementation of \cb must block after running out of memory.
By not enforcing the delivery of \emph{old} messages, \tcb's tail-validity lets processes recycle the memory dedicated to these messages.
This is why \tcb requires only finite memory while \cb does not. In Section~\ref{sec:consensus}, we show that despite its weaker semantics, \tcb is sufficient for solving consensus.

\subsection{Algorithm} \label{sec:tcb:algorithm}

Algorithm~\ref{alg:tcb} implements \tcb using finite memory
with a fast/slow path approach that avoids both signatures and disaggregated memory in the common case.
For pedagogical reasons, we assume that a designated process is the broadcaster while the others are receivers.
Each receiver owns an array of $t$ Single-Writer Multiple-Reader (SWMR) regular registers.
Each register is only writable by its owner, but is readable by all processes.
The regularity of the registers forces READs that execute concurrently to a WRITE to return either the value being written or the previous one.
Moreover, each process uses a \tbfull (\tb) primitive which ensures the delivery by correct processes of the last $2t$ messages broadcast through it, 
but does not prevent equivocation.
Formally, \tb has all properties of \tcb except agreement.

Implementing \tb using finite memory is straightforward.
Briefly, the broadcaster buffers its last $2t$ messages and retransmits them until it receives acknowledgements from all receivers.
Then, to broadcast a new message when the buffer is full, the broadcaster simply makes room for it by evicting the oldest buffered message.

\begin{lstlisting}[float=ht!,caption={\tcbfull},label={alg:tcb}]
@\textit{\# at the broadcaster}@
def broadcast(k, m):
  TBcast-broadcast <LOCK, k, m> @\label{alg:tcb:lock}@
  TBcast-broadcast <SIGNED, k, m, sign((k, m))> @\label{alg:tcb:signed}@

@\textit{\# at receivers}@:
SWMR[me] = @$[(-1,\ \bot,\ \bot),\ \ldots \ ]$@ @\textit{\# array of t slots}@
delivered = @$[-1,\ \ldots \ ]$@ @\textit{\# array of t slots}@ @\label{alg:tcb:delivered-array}@
locks = @$[(-1,\ \bot),\ \ldots \ ]$@ @\textit{\# array of t slots}@ @\label{alg:tcb:locks-array}@
locked = @$[[(-1,\ \bot),\ \ldots \ ],\ \ldots \ ]$@ @\textit{\# array of t*n slots}@ @\label{alg:tcb:locked-array}@

upon TBcast-deliver <LOCK, k, m> from p: @\label{alg:tcb:upon-lock}@
  k', _ = locks[k%t] @\label{alg:tcb:read-lock}@
  if k > k':@\label{alg:tcb:check-higher-lock}@
    locks[k%t] = (k, m)@\label{alg:tcb:update-lock}@
    TBcast-broadcast <LOCKED, k, m> @\label{alg:tcb:locked}@

upon TBcast-deliver <LOCKED, k, m> from q:@$\label{alg:tcb:upon-locked}$@
  k', _ = locked[q][k%t]
  if k > k':
    locked[q][k%t] = (k, m)
    if locked[r@$_0$@][k%t] == ... == locked[r@$_{n-1}$@][k%t]:@\label{alg:tcb:unanimity}@
        deliver_once(k, m) @\label{alg:tcb:fast-delivery}@
    
upon TBcast-deliver <SIGNED, k, m, sig> from p: @\label{alg:tcb:upon-signed}@
  if valid(sig, (k, m), p):@\label{alg:tcb:valid-sig}@
    k', m' = locks[k%t]
    if k > k' or k == k' and m == m': @\label{alg:tcb:check-higher}@
      locks[k%t] = (k, m) @\label{alg:tcb:update-lock-signed}@
      SWMR[me][k%t].write((k, sig, m)) @\label{alg:tcb:copy}@
      for each (k', s', m') in SWMR[*][k%t]: @\label{alg:tcb:for}@
        if valid(s', (k', m'), p): @\label{alg:tcb:valid-check}@
          if k' == k and m' != m: @\label{alg:tcb:byzantine-check}@
            return @\textit{\# Byzantine broadcaster} \label{alg:tcb:byzantine-return}@
          if k' @$>$@ k and k' @$\equiv$@ k (mod t): @\label{alg:tcb:tail-check}@
            return @\textit{\# out of tail}@
      deliver_once(k, m) @\label{alg:tcb:slow-delivery}@
        
def deliver_once(k, m):
  if k > delivered[k%t]:@\label{alg:tcb:if-delivered}@
    delivered[k%t] = k@\label{alg:tcb:update-delivered}@
    trigger deliver(k, m)
\end{lstlisting}

As mentioned, Algorithm~\ref{alg:tcb} has a low-latency fast path that avoids signatures and disaggregated memory.
It also incorporates a fallback slow path for liveness.
For presentation simplicity, it triggers the slow path in parallel to the fast path (lines~\ref{alg:tcb:lock} and~\ref{alg:tcb:signed}), but in reality, \sysname triggers it when replicas fail to decide on new client requests after some configurable timeout.
In addition to the shared SWMR registers, receivers use three finite-size local arrays for bookkeeping (lines~\ref{alg:tcb:delivered-array}-\ref{alg:tcb:locked-array}).

In the fast path, the broadcaster first \tb-broadcasts its message alongside its identifier within a \texttt{LOCK} message (line~\ref{alg:tcb:lock}).
When receivers \tb-deliver this message (line~\ref{alg:tcb:upon-lock}),
they commit not to deliver any other message for the given identifier, and tell other receivers about their commitment by \tb-broadcasting a \texttt{LOCKED} message (line~\ref{alg:tcb:locked}).
Receivers use \texttt{locks} (lines~\ref{alg:tcb:read-lock}-\ref{alg:tcb:update-lock}) to avoid committing to different messages for the same identifier.
Importantly, receivers store only up to $t$ commitments in this array, by evicting earlier commitments that alias to the same index $k\%t$ (line~\ref{alg:tcb:update-lock}).
When receivers learn that everyone committed to the same message (lines~\ref{alg:tcb:upon-locked}-\ref{alg:tcb:unanimity}), they know that no correct replica will deliver a different message and thus deliver it (line~\ref{alg:tcb:fast-delivery}).

In the slow path, the broadcaster additionally \tb-broadcasts a signed version of its message (line~\ref{alg:tcb:signed}).
After \tb-delivering a signed message (line~\ref{alg:tcb:upon-signed}),
receivers verify its signature (line~\ref{alg:tcb:valid-sig}).
Then, they check that they have not committed to a different message for the same $k$ (line~\ref{alg:tcb:check-higher}), and ensure that they will not do so in the future (line~\ref{alg:tcb:update-lock-signed}).
Subsequently, they copy the signed message to their SWMR register associated with the message identifier $k$ (line~\ref{alg:tcb:copy}), before reading the associated SWMRs owned by other receivers (line~\ref{alg:tcb:for}).
Receivers ignore messages with invalid signatures (line~\ref{alg:tcb:valid-check}) and abort delivery if they detect a different message for the same identifier (line~\ref{alg:tcb:byzantine-check}).
In case receivers detect another message with a higher identifier that is associated with the same SWMR registers (line~\ref{alg:tcb:tail-check}), they drop their own message as it no longer belongs to the tail.
Otherwise, they deliver it (line~\ref{alg:tcb:slow-delivery}).

The correctness of the slow path of Algorithm~\ref{alg:tcb} hinges on that all correct processes will find the message copied by the \emph{fastest} correct replica when reading the registers.
Thus, they can deliver no message other than the one this first replica copied, thereby preserving agreement.
The fast and slow paths are linked together via the \texttt{locks} array (lines~\ref{alg:tcb:update-lock} and ~\ref{alg:tcb:update-lock-signed}), which ensures that whichever path executes first forces the value of the message for the other path.
Note that messages broadcast by Byzantine processes might never be delivered by correct processes, as allowed by the specification. A more detailed proof of correctness is given in 
\ifextended
Appendix~\ref{app:correct-TCB}.
\fi
\section{State-Machine Replication} \label{sec:consensus}

Like many leader-based BFT consensus protocols,
  \sysname's protocol has the same high-level
  layout as PBFT~\cite{castro1999practical}: it shares
  naming conventions and
  splits the protocol in similar phases.
However, \sysname has different goals
  ($2f{+}1$ fault tolerance, finite memory, microsecond
  latency), different assumptions (disaggregated
  memory), and a different non-equivocation mechanism
  (\tcbfull (\tcb)).
Consequently, while the protocol retains a PBFT-like
    conceptual organization, its internal structure and
    mechanics differ substantially from those of PBFT, as we
    describe in this section.

We start by giving an overview of our consensus protocol.
Then, we analyse its slow path and explain how it deals with finite memory and view changes.
Finally, we describe how its fast path leverages unanimous and timely collaboration of replicas to achieve microsecond-scale latency.
For presentation clarity, this section gives only an informal description of our protocol. Its pseudocode is given in
\ifextended
Appendix~\ref{app:correct-consensus} alongside detailed arguments of its correctness.
\fi

\subsection{Basic Protocol}

\begin{figure}
    \centering
    \includegraphics[width=\columnwidth]{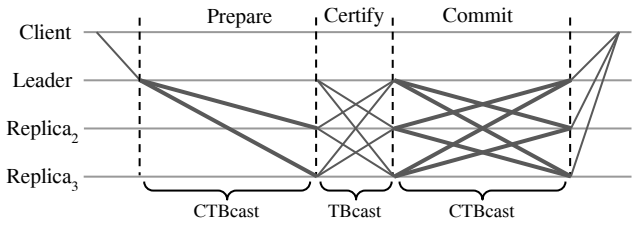}
    \caption{Communication pattern of \sysname's slow path. Bold lines represent messages sent over \tcb. Thinner lines represent direct messages.}
    \label{fig:consensus-slow-path}
\end{figure}

From a high-level point of view, the slow path of our consensus protocol---shown in Figure~\ref{fig:consensus-slow-path}---has three phases: \emph{Prepare}, \emph{Certify}, and \emph{Commit}.
After the leader receives a signed request from a client, it \emph{proposes} it by broadcasting a \texttt{PREPARE} message via \tcb.
When replicas (including the broadcaster) deliver this message, they proceed to its certification: each replica (1) signs the proposal, (2) \tb-broadcasts its signature in a \texttt{CERTIFY} message, and (3) aggregates $f{+}1$ signatures into an unforgeable proof that the proposal was emitted by the leader.
This proof is critical during a leader change~(\S\ref{sec:view-change}) to force the new leader to propose no other command than one that might have already been applied.
The goal of the subsequent Commit phase is to guarantee that, if a request is applied, a proof of its proposal exists and will survive a leader change (unless garbage collected beforehand,~\S\ref{sec:smr-non-eq}).
To achieve this, replicas first \tcb-broadcast a \texttt{COMMIT} message containing the proof.
Once a replica delivers $f{+}1$ such \texttt{COMMIT} messages,
it knows that no future leader can emit a conflicting proposal~(\S\ref{sec:view-change}), and can safely apply the client's request to its local state machine and respond to the client.
Note that we use \tb instead of the more expensive \tcb for broadcasting \texttt{CERTIFY} messages: because each certificate involves at least one correct replica and all correct replicas certify the same proposal---thanks to \tcb in the Prepare phase---it is impossible to build certificates for different proposals, even if some Byzantine replicas equivocate \texttt{CERTIFY} messages.

The protocol described so far replicates a single client request.
Similarly to PBFT, \sysname uses a sliding window to run its consensus protocol across a series of slots.
As the leader receives multiple requests, it proposes each in a different slot, enabling many slots to progress in parallel.
\sysname also adopts PBFT's application checkpoints to limit the number of concurrent proposals.
This mechanism constrains the influence of a Byzantine leader and bounds the number of relevant messages at any time.
An application checkpoint is signed by $f{+}1$ replicas and includes (1) the application state after applying the first $i$ ordered requests, and (2) an implicit authorization to work on slots $[i+1, i+\t{window}]$.

\subsection{Non-Equivocation at the Consensus Level}
\label{sec:smr-non-eq}

An important aspect of \sysname's consensus protocol is how it prevents the leader from issuing conflicting proposals.
\tcb only partially addresses this problem:
it prevents a Byzantine leader from sending conflicting proposals with the same \tcb identifier, but it does not prevent conflicting proposals from being spread using different \tcb identifiers.

Conceptually, a process can verify that another process has not equivocated only if it knows the complete history of messages broadcast by that process.
For this reason, our protocol requires processes to interpret messages from other processes in FIFO order.
However, \tcb does not guarantee delivery of all broadcast messages because of its tail-validity property, which may prevent processes from delivering \emph{all} messages from a correct process in FIFO order.
We address this limitation by pairing \tcb with \emph{\tcb \summaries}.

A \tcb \summary is an unforgeable synopsis of the messages broadcast by a process $p$ via \tcb up to a given \tcb identifier.
Summaries serve as certificates signed by $f{+}1$ replicas that have witnessed $p$'s broadcasts and attest that $p$ has not equivocated at the consensus level.
Once a process receives a \summary for $p$ up to some \tcb identifier $i$, it can safely apply it and continue processing $p$'s subsequent messages with identifiers greater than $i$ in FIFO order.
In essence, \tcb \summaries restore FIFO delivery guarantees that may be violated due to the tail-validity of \tcb.

\tcb \summaries are generated interactively.
For every $t$ consecutive \tcb messages that a replica $r$ delivers from $p$ ($t$ is \tcb's tail parameter), $r$ participates in creating a \summary that captures the state of $p$.
Replica $p$ tracks how many messages it has broadcast and blocks, every $t$ messages, waiting for a \summary of its own state.\footnote{To avoid blocking-induced latency hiccups, our implementation generates summaries every $t/2$ messages, effectively double-buffering the tail so that broadcasting can continue while summaries are being produced.}
Using its previous \summary, $p$ can bring correct processes that missed some of its messages up to date, and help them deliver the last $t$ messages it has broadcast so far.
Consequently, even if up to $f$ Byzantine replicas fail to assist in building $p$'s next \summary, all correct replicas will eventually collaborate to generate it.
This guarantees the liveness of the \summaries and allows $p$ to continue broadcasting.

Ensuring that \tcb \summaries have bounded size is crucial.
To achieve this, each process retains only a limited number of messages that it has \tcb-delivered from other processes.
This bound is enforced by maintaining, for every broadcaster, a sliding window of consensus slots.
\sysname requires processes to \tcb-broadcast application checkpoints, which allow receivers to advance these windows.
When a receiver slides a broadcaster's window forward, it discards all messages that refer to slots falling outside the window, as the corresponding consensus slots have already been checkpointed.
As a result, the relevant state of a broadcaster consists solely of a finite consensus window range and the finite collection of \tcb messages that fall within it.

\subsection{View Change}
\label{sec:view-change}

To tolerate faulty leaders, \sysname uses a PBFT-inspired \emph{view change} mechanism.
Execution is divided into views, each with a designated leader selected in a round-robin fashion.
When a replica observes no progress or suspects the leader of censoring requests, it
moves to the next view by \tcb-broadcasting a \texttt{SEAL\_VIEW} message.

Before the leader of a view can propose new requests, it must transfer requests that may have been applied in previous views.
It does so by \tcb-broadcasting a \texttt{NEW\_VIEW} message that summarizes the system state after the last stable checkpoint.

The \texttt{NEW\_VIEW} message contains the latest application checkpoint and, for each open consensus slot following this checkpoint, the relevant \texttt{COMMIT} messages from $f{+}1$ replicas.
This information is assembled as follows.
When a replica $p$ receives a \texttt{SEAL\_VIEW} message from replica $q$, $p$ generates a certificate share attesting to $q$'s state.
This state includes $q$'s latest checkpoint and the most recent \texttt{COMMIT} message sent by $q$ for each of its open slots.
The leader collects $f{+}1$ matching certificate shares for $f{+}1$ distinct replicas and includes them in the \texttt{NEW\_VIEW} message.

The \texttt{NEW\_VIEW} message constrains the proposals of the new leader.
For each open slot, the leader must propose the request contained in the \texttt{COMMIT} message with the highest view number, if such a message exists.
Only slots without any \texttt{COMMIT} message may be assigned new requests.

This mechanism prevents a faulty leader from overwriting previously applied requests.
If a correct replica applied a request in some view, it must have received $f{+}1$ matching \texttt{COMMIT} messages.
Consequently, when the new leader gathers certificates about $f{+}1$ replicas, at least one certificate must contain a \texttt{COMMIT} message for each applied request.
As a result, no applied request can be overwritten in the new view.

\subsection{Fast Path}
\label{sec:fast-path}

\begin{figure}
    \centering
    \includegraphics[width=\columnwidth]{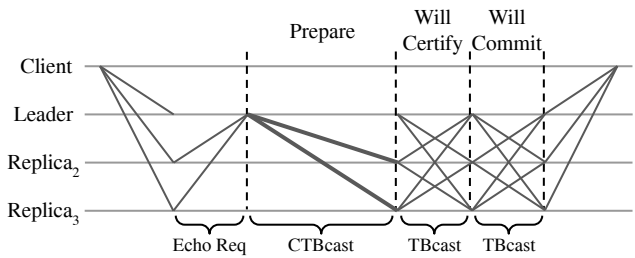}
    \caption{Communication pattern of \sysname's fast path. Bold lines represent messages sent over \tcb. Thinner lines represent direct messages.}
    \label{fig:consensus-fast-path}
\end{figure}

To operate at the microsecond scale, 
\sysname incorporates a fast path---shown in Figure~\ref{fig:consensus-fast-path}---that moves signatures out of the critical path in times of synchrony and unanimous collaboration.
Similar to the slow path, the fast path has three phases: Prepare, which is common with the slow path but executes the fast path of \tcb, followed by \textit{WillCertify} and \textit{WillCommit}.
The last two phases replace the Certify and Commit phases of the slow path with inexpensive rounds of \tbfull.

\sysname's fast path is simple.
For a given consensus slot, after the end of the Prepare phase,
replicas broadcast a \texttt{WILL\_CERTIFY} message and wait to receive the same message from all others.
Once received, they proceed with broadcasting a \texttt{WILL\_COMMIT} message and again wait for unanimity before applying the proposed request.
Both messages contain solely the view number and the consensus slot.
These messages are essentially promises that their broadcasters will run the slow path before \tcb-broadcasting their next \texttt{SEAL\_VIEW} messages.
By broadcasting \texttt{WILL\_CERTIFY}, a replica promises to participate in certifying the \texttt{PREPARE} message.
With \texttt{WILL\_COMMIT}, it promises to \tcb-broadcast the resulting certificate within a \texttt{COMMIT} message.

The safety of this scheme is intuitive.
If a replica receives $2f{+}1$ \texttt{WILL\_CERTIFY} messages, it knows that at least $f{+}1$ correct replicas will certify the \texttt{PREPARE} message in the Certify phase.
Similarly, receiving $2f{+}1$ \texttt{WILL\_COMMIT} messages means that at least $f{+}1$ correct replicas will send a \texttt{COMMIT} message in the Commit phase and thus no other request will be applied for this given slot.

The fast path also takes care of finite memory as it does not keep promises forever: it drops the promises that refer to consensus slots included in an application checkpoint.
These checkpoints, along with the \tcb \summaries, make up the required background signatures of \sysname's fast path.

Lastly, the fast path eschews signatures between clients and replicas by having clients send unsigned requests to all replicas, rather than a signed request to the leader.
A replica endorses a \texttt{PREPARE} message for a client request only if it has also received that request directly from the client.
However, a Byzantine client could send its request only to the leader, causing the proposal to stall and potentially triggering a view change.
To prevent this scenario, the fast path introduces an additional communication round (denoted \textit{Echo Req} in Figure~\ref{fig:consensus-fast-path}), in which the leader waits for followers to echo the client request before proposing it.
This ensures that followers have received the request and will participate in certifying the corresponding \texttt{PREPARE} message, provided it is delivered in a timely manner.
\section{Implementation}\label{sec:implementation}
Our implementation of \sysname leverages disaggregated memory to build the shared registers used by \tcbfull, and a fast networking fabric to build point-to-point links with microsecond-scale latency.
This section explains how we use RDMA to achieve these goals.

\subsection{Reliable SWMR Regular Registers}
\label{sec:implementation-registers}

The \tcbfull primitive used by \sysname requires reliable, SWMR, regular registers in its slow path.
\textit{Reliable} means that the registers do
not fail, i.e., \texttt{READ}s and \texttt{WRITE}s always complete.
\textit{SWMR} means Single-Writer Multiple-Reader:
each register has an owner, which is a replica that is
allowed to write to it, while all other replicas
may only read the register.
\textit{Regular} means that, when a \texttt{READ} executes concurrently to a \texttt{WRITE}, the former should return either the value that is being written or the previous one.
We explain below how to
  build registers with such properties
  on top of memory nodes exposing their
  memory over RDMA (Figure~\ref{fig:swmr}).

\paragraph{SWMR Register.}
We implement SWMR semantics using RDMA permissions.
RDMA splits memory into regions, each with different access permissions based on tokens.
We simply create a read-write and a read-only region for
  the same memory range, give the read-write token
  to the writer of the register, and the read-only
  token to the other replicas.

\begin{figure}
    \centering
    \includegraphics[width=1.0\columnwidth]{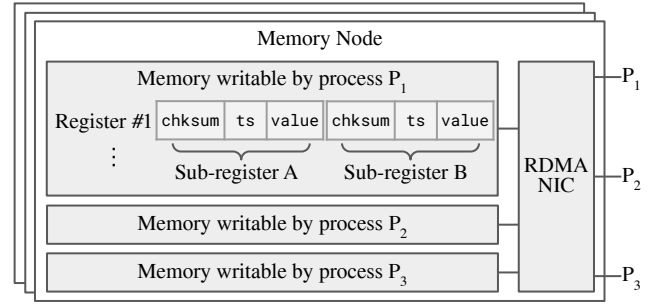}
    \caption{Reliable SWMR Regular registers using RDMA.}
    \label{fig:swmr}
\end{figure}

\paragraph{Regular Register.}
RDMA-exposed memory is atomic (hence regular), but
  only at an 8-byte granularity.
Thus, an RDMA \texttt{READ} concurrent with a \texttt{WRITE} may return partially written data,
  mixing old and new values.
To detect this problem, we use
checksums, as in Pilaf~\cite{mitchell2013using}.
A simple approach, where a reader retries until
  the checksum is valid, violates liveness as
a Byzantine writer can write bogus checksums.
To avoid such scenario, we follow an evolved double-buffering strategy, which ensures that a reader is always able to find a complete \texttt{WRITE} or detect the owner of the register as being Byzantine, as explained below.
As depicted in Figure~\ref{fig:swmr}, each register is made of two sub-registers.
Each \texttt{WRITE} to a given register is directed to one of the sub-registers in a round-robin manner.
To write a value, the writer prefixes it with a logical timestamp (denoted \texttt{ts}) and a checksum. 
Importantly, the writer waits for $\delta$ (the known communication bound after GST) between two \texttt{WRITES} to the same register to give readers time to find a complete \texttt{WRITE}.
To perform a \texttt{READ}, the reader reads both sub-registers at once and, out of the values with a valid checksum, returns the one with the highest timestamp.
\footnote{
The hardware is allowed to reorder RDMA \texttt{READ}s following RDMA \texttt{WRITE}s when issued to different Queue Pairs (QPs)~\cite{rdma-ordering}.
To ensure regularity, i.e., that subsequent register \texttt{READ}s see the RDMA-written value, a register \texttt{WRITE} only returns after the PCIe \texttt{WRITE} transaction reaches the last-level cache (L3).
We do so by issuing an RDMA \texttt{READ} after the RDMA \texttt{WRITE} to the same QP---which acts as a PCIe fence~\cite{rdma-read-flush}---and only considering that the register \texttt{WRITE} completes when the RDMA \texttt{READ} completes.
}
If both checksums are invalid and the \texttt{READ} 
  took less than $\delta$, then the writer is Byzantine (i.e., it did not respect the $\delta$ cooldown, or wrote bogus data), so a default value is returned.
  The writer is also deemed Byzantine if both sub-registers have the same timestamp.
If the \texttt{READ} took more than $\delta$, it is retried.
This scheme works in eventually synchronous systems with bounded clock drift~\cite{ukharon, farm2}: after GST~(\S\ref{sec:background:model}), RDMA \texttt{READ}s take less than $\delta$, so they cannot overlap with \texttt{WRITE}s on both sub-registers, so at least one sub-register will eventually have a valid checksum (assuming the writer is correct), and timestamp-based ordering ensures regular register semantics.

\paragraph{Reliable Register.} 
We replicate each register to $2f_m{+}1$ memory nodes (Figure~\ref{fig:swmr}).
Here, $f_m$ is the maximum number of memory
  nodes that may crash.
While memory nodes add to the total number of
  replicas in the system, these nodes
  do not replicate the application, and they can be shared among many replicated applications, as each application takes little memory (\S\ref{eval:mem}).
Our register replication scheme is straightforward.
\texttt{WRITE}s are issued to all memory nodes in parallel and return after having completed at $f_m{+}1$ of them.
\texttt{READ}s are also issued to all nodes in parallel, wait for $f_m{+}1$ of them to complete, and return the value of the regular register with the highest timestamp.
Because \texttt{READ}s and \texttt{WRITE}s always complete at a majority, this scheme is trivially live.
Moreover, intersecting quorums guarantee that \texttt{READ}s intersect with the last completed \texttt{WRITE} and/or with a concurrent \texttt{WRITE} (if any) at some register(s).
Thus, \texttt{READ}s return a value no older than the value of the last completed \texttt{WRITE}, thereby preserving regularity.

\subsection{A Fast Message-Passing Primitive}

\begin{figure}
    \centering
    \includegraphics[width=\columnwidth]{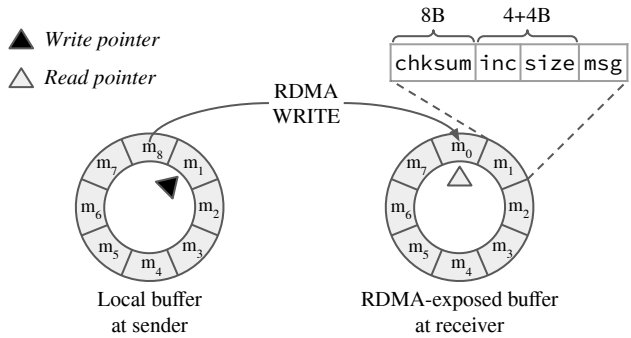}
    \caption{\sysname's RDMA-powered message-passing primitive.}
    \label{fig:p2p}
\end{figure}

To achieve microsecond-scale communication, \sysname implements a fast, one-way messaging primitive between a sender and a receiver where, similarly to \tcb, the receiver is required to deliver only the last $t$ messages sent. This primitive admits an implementation without receiver acknowledgements, which we found to be important for achieving microsecond-scale performance.

Figure~\ref{fig:p2p} depicts the implementation of this primitive over RDMA.
The receiver has a circular buffer
  exposed over RDMA; the buffer is divided into $t$ slots of equal size large enough for the largest message.
Briefly, the sender RDMA-writes messages to the receiver's buffer, while the latter scans its local buffer for new messages.
There are no acknowledgments:
the sender overwrites old messages with newer ones, even if they were not yet delivered.

We now explain this implementation in detail.
The sender allocates a mirror image of the receiver's buffer in its local memory, and maintains locally a \emph{write pointer} to the slot for its next message.
Each slot has a header composed of a checksum, an incarnation number (the number of times it was written), and a message size.
To send a new message, the sender writes it to the slot pointed by the writer pointer and fills its header.
Then, it issues an RDMA \texttt{WRITE} to the corresponding slot in the receiver's memory, and marks the slot as unavailable until it is notified of the completion of the \texttt{WRITE} by the RDMA NIC.
Finally, the sender advances its write pointer.
If the pointed slot is unavailable, the new message is queued in a second (not depicted) circular buffer.
This buffer acts as a staging area: it forwards its messages for transmission when slots become available, and evicts the oldest queued message to accommodate a newer one.

The receiver maintains a \emph{read pointer} to the slot where it will read the next message.
The receiver polls this slot for a particular incarnation number, which identifies the next message it expects to find.
Once this incarnation number is seen, the receiver copies the entire message to a private buffer in order to avoid interfering \texttt{WRITE}s on the same slot.
Then, the receiver checks the incarnation number again in the copied message.
If the incarnation number has not changed, the sender verifies the checksum before delivering the message (rescheduling the polling if the checksum is invalid).
If the receiver finds a higher incarnation number than expected, 
it concludes that some older messages may exist in other slots of the buffer that will have to be delivered first.
So, the receiver aborts the delivery and advances its pointer to the oldest undelivered message.
With this strategy, the receiver guarantees FIFO delivery of the last $t$ messages.

This scheme has two benefits: it uses practically bounded memory and avoids acknowledgements.
The latter---even when batched---increase the application's tail latency as scheduling an acknowledgement alone takes $\approx$300\,ns~\cite{kalia2014using}, which is time lost handling incoming events.
Instead, by the End-to-End Principle~\cite{end-to-end}, acknowledgements are piggybacked in SMR-level messages.

\section{Evaluation} \label{sec:evaluation}

We evaluate the various performance traits of \sysname and verify its suitability as a BFT SMR system for microsecond applications. 
We aim to answer the following:
\begin{itemize}
    \item How much latency does \sysname induce on the replicated applications (\S\ref{eval:eelat})?
    \item How does the replication latency of \sysname compare to other SMR systems (\S\ref{eval:replat})?
    \item How do the internal components of \sysname contribute to its end-to-end latency (\S\ref{eval:latbreak})?
    \item How does our implementation of \tcb perform in comparison to SGX-based non-equivocation mechanisms (\S\ref{sec:exp-non-eq})?
    \item How does the tail parameter of \tcb impact \sysname's tail-latency (\S\ref{sec:tail-eval-latency})?
    \item What is the memory consumption of \sysname (\S\ref{eval:mem})?
\end{itemize}

\begin{table}
 \small
 \centering
    \caption{Configuration details of machines.}
\label{table:hwspecs}
	\begin{tabular}{cl}
	\toprule
\textbf{CPU}		& 2$\times$ 8c/16t Intel Xeon Gold 6244 @ 3.60\,GHz \\
\textbf{NIC/Switch}		& Mlnx CX-6 MT28908 / MSB7700 EDR 100 Gbps \\
\textbf{Software} & Linux 5.4.0-74-generic / Mlnx OFED 5.3-1.0.0.1 \\
	 \bottomrule
	\end{tabular}
\end{table}

Our testbed is a cluster with 4 servers configured per Table~\ref{table:hwspecs}.
The dual-socket machines have an RDMA NIC attached to their first socket.
Our experiments execute on cores of the first socket using local NUMA memory.
Our implementation~\cite{ubft-zenodo} measure time using TSC~\cite{tsc} via \t{clock\_gettime} with the \t{CLOCK\_MONOTONIC} parameter.
All machines are connected to a single switch.
Due to limited hardware availability, servers act as both compute and memory nodes.

In all experiments, we deploy 1 client and 3 replicas, we replicate SWMR registers to 3 memory nodes, and we take at least 10,000 measurements.
Additionally, we set the consensus window to 256 requests and---unless stated otherwise---the tail parameter of \tcb to 128 messages.

\paragraph{Applications.}
We integrate \sysname with MemCached~\cite{kalia2014using}, Redis~\cite{redis} and Liquibook~\cite{liquibook}.
MemCached and Redis are non-replicated high-performance key-value stores.
Liquibook is a financial order matching engine.
We also integrate all the aforementioned applications with Mu~\cite{aguilera2020microsecond}, the SMR system the lowest replication latency to date but that tolerates only crash faults.
In all applications, the client sends messages using \sysname's RPC mechanism.
Additionally, using a no-op application, we compare \sysname against MinBFT~\cite{VeroneseCBLV13}, a state-of-the-art $2f{+}1$ BFT SMR system with
a publicly available implementation based on Intel's SGX~\cite{sgx-explained}.
SGX provides a secure CPU enclave for executing arbitrary code,
thus offering a general-purpose trusted computing base.

\paragraph{Implementation Effort.}
We implemented \sysname on top of a framework~\cite{ukharon} that facilitates the use of RDMA.
Our prototype spans 11,750 lines of C++, out of which 2,966 are dedicated to \tcbfull.
The prototype includes all features on the critical path of a complete implementation: the only major unimplemented features are application and replica state transfers.
We use Dalek's implementation of EdDSA~\cite{dalek} for public-key cryptography, BLAKE3~\cite{blake3} for HMACs, and xxHash~\cite{xxhash} for checksums.

\subsection{End-to-End Application Latency}
\label{eval:eelat}

\begin{figure}
    \centering
    \includegraphics[width=\columnwidth]{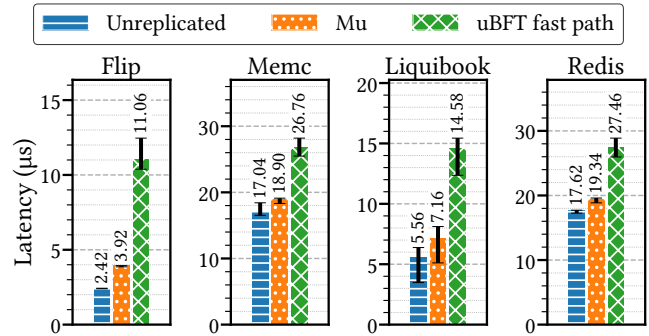}
    \caption{End-to-end latency of different applications when either not replicated or replicated via Mu and \sysname's fast path. Printed values show the 90th percentiles; whiskers indicate the 50th and 95th percentiles.}
    \label{fig:replicated-apps-latency}
\end{figure}

We first explore the replication overhead that \sysname induces to end applications.
We compare the latency of its fast path against unreplicated execution and replication via Mu.
We study four applications: Flip, a toy application that reverses its input, as well as Memcached, Redis, and Liquibook.

Flip uses 32\,B requests and responses. The key-value stores use 16\,B keys and 32\,B values; their workloads consist of 30\% \texttt{GET} operations, 80\% of which return a value. Liquibook requests are 32\,B, and responses range from 32\,B to 288\,B depending on the number of matched orders; 50\% of requests are \texttt{BUY} operations and 50\% are \texttt{SELL} operations.

Figure~\ref{fig:replicated-apps-latency} shows the results. \sysname is approximately 7.5\,\us slower than Mu at the 90th percentile across applications. The relative overhead is largest for ultra-low-latency apps such as Flip, where \sysname is about $3\times$ slower than Mu. As application latency increases, the relative overhead decreases: \sysname is about $2\times$ slower for Liquibook and roughly $1.5\times$ slower for the key-value stores.

\sysname also slightly increases latency variance (the gap between the 50th and 95th percentiles) compared to Mu. This is primarily due to additional communication required by \sysname: the RPC layer performs an extra round to ensure that all correct replicas receive the client request, and the replication protocol requires four broadcast rounds before replying to the client, compared to a single majority \texttt{WRITE} in Mu. Overall, \sysname's fast path introduces four additional communication rounds relative to Mu, which explains the higher tail latency. Nevertheless, the additional variance remains below 3\,\us.

Overall, \sysname enables microsecond-scale applications to obtain Byzantine fault tolerance with at most about 10\,\us of overhead when the network is timely and replicas operate synchronously.

\subsection{End-to-End Replication Latency}
\label{eval:replat}

We evaluate \sysname's replication latency across varying request sizes using a no-op application---whose reply size matches the request---when replicated via \sysname. We compare against three other configurations: no replication, Mu, and MinBFT, an SGX-based $2f{+}1$-replica BFT alternative.

Figure~\ref{fig:noop-latency} shows the median end-to-end latency.
As expected, the lowest latency is achieved without replication (denoted \textit{Unrepl.}). In this configuration, end-to-end latency ranges from 2.2\,\us to 20\,\us and primarily reflects the cost of communicating with the server using our RPC mechanism.

\begin{figure}
    \centering
    \includegraphics[width=\columnwidth]{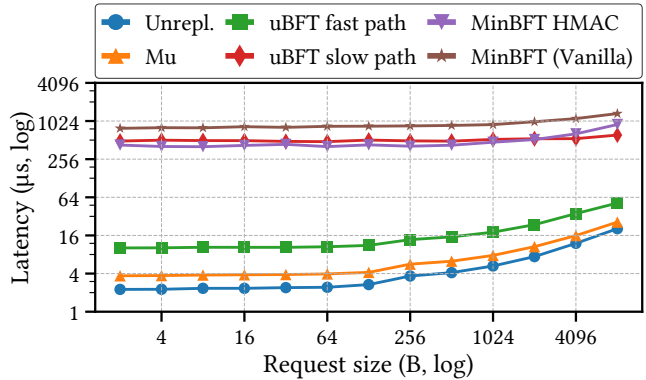}
    \caption{Median end-to-end latency for different request sizes of an unreplicated no-op application, as well as its latency when replicated with Mu, \sysname and MinBFT.}
    \label{fig:noop-latency}
\end{figure}

Replicating via Mu increases end-to-end latency by up to 64\% for small requests and by at most 26\% for 8\,KiB requests. In the absence of failures, Mu's leader replicates requests by RDMA-writing them to its followers.

\sysname's fast path exhibits higher latency than Mu due to four additional communication rounds, but increases latency by at most 175\% relative to Mu while providing Byzantine fault tolerance.

One might expect \sysname's fast path to come at the cost of high slow-path latency; however, this is not the case. We compare our slow path against MinBFT in two configurations. In its vanilla configuration, MinBFT uses HMACs between replicas, while clients sign requests using public-key cryptography, leading to a minimum end-to-end latency of 566\,\us. We modify MinBFT to also use SGX at the clients, replacing public-key cryptography with HMACs. As a result, \sysname's slow path is faster than vanilla MinBFT (by up to 52\%) and at most 24\% slower than the HMAC-only variant, despite relying on public-key cryptography.

Note that MinBFT is not RDMA-optimized and uses TCP in its implementation. To improve comparability, we used Mellanox's VMA library~\cite{vma} to replace MinBFT's TCP stack with a kernel-bypass alternative that leverages RDMA NICs. Moreover, because our setup does not support SGX, we emulated it using the latency results from Section~\ref{sec:exp-non-eq}.

Overall, these results show that \sysname achieves Byzantine fault tolerance with moderate latency overhead on the fast path while maintaining competitive slow-path performance.

\subsection{Latency Breakdown}
\label{eval:latbreak}

\begin{figure}
    \centering
    \includegraphics[width=\columnwidth]{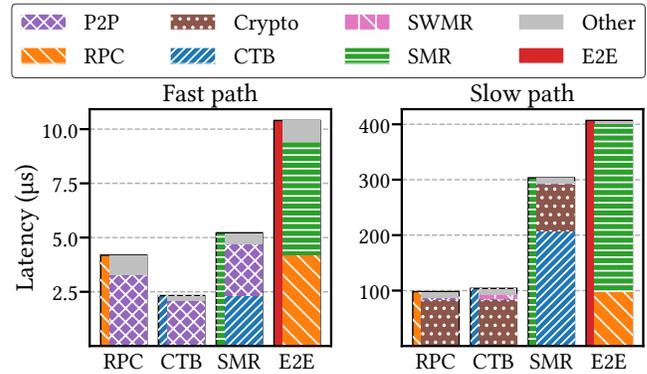}
    \caption{Recursive decomposition of the end-to-end latency of \sysname's fast and slow path when replicating Flip with requests of 8\,B.}
    \label{fig:stacked-latency}
\end{figure}

To better understand \sysname's latency, we break down the latency of its internal components.

Figure~\ref{fig:stacked-latency} shows a recursive decomposition of the latency of an 8\,B Flip request into its main components: remote-procedure call, \tcbfull, and replication (denoted \textit{\texttt{RPC}}, \textit{\texttt{CTB}}, and \textit{\texttt{SMR}}, respectively). The rightmost columns show the end-to-end client-perceived latency (denoted \textit{\texttt{E2E}}). Each bar is split into two regions: the narrow left region shows the total latency of the component, while the wider right region shows its internal decomposition.

We identify four primitive sources of latency. The first is communication over our point-to-point messaging primitive (denoted \textit{\texttt{P2P}}). The second is signature generation and verification (denoted \textit{\texttt{Crypto}}),\footnote{The \texttt{Crypto} category includes not only cryptographic computation but also synchronization costs, such as dispatching operations to a thread pool and retrieving the results.} and the third is accessing disaggregated memory registers (denoted \textit{\texttt{SMWR}}). The latter two are only relevant in \sysname's slow path. Finally, the \textit{\texttt{Other}} category captures glue logic between components, including buffer copies and delays between the arrival and processing of asynchronous events.

In the fast path, most of the latency is due to communication. Given the small message sizes, improving end-to-end latency would primarily require either reducing the number of communication steps or lowering network latency.

In the slow path, public-key cryptography dominates latency. Because signatures are unavoidable in our setting~\cite{aguileraBGPXZ21}, improving slow-path latency largely depends on faster cryptographic primitives. Conversely, the cost of accessing disaggregated memory (part of \texttt{CTB}) is comparatively small, accounting for 14\,\us (3.5\%) of the end-to-end latency.

\subsection{Latency of Non-Equivocation}\label{sec:exp-non-eq}

Figure~\ref{fig:stacked-latency} shows that \tcbfull (\tcb) accounts for a substantial portion of the latency of \sysname's fast and slow paths.
This section quantifies the performance of our \tcb implementation against the de facto method for preventing equivocation in modern systems~\cite{VeroneseCBLV13,behl2017hybrids,fastbft}, namely a trusted counter implemented on Intel SGX.

\begin{figure}
    \centering
    \includegraphics[width=\columnwidth]{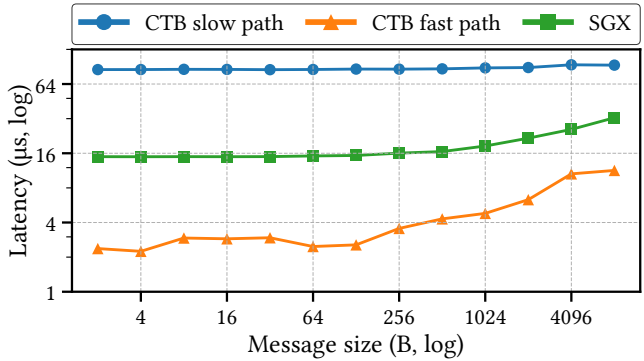}
    \caption{Median latency of multiple non-equivocation mechanisms for different message sizes.}
    \label{fig:anti-equivocation-latency}
\end{figure}

Non-equivocation mechanisms based on trusted counters securely bind a monotonically increasing sequence number to each broadcast message. 
Briefly, each enclave stores a local $counter$ and a shared $secret$. Before sending a message, a process feeds it to its enclave, which returns a proof of non-equivocation of the form $HMAC_{secret}(msg \| counter{+}{+} \| process~id)$. 
Recipients verify the HMAC using their own enclave, ensuring that each message is uniquely identified and preventing equivocation. 
Because authentication occurs entirely within the enclave, only symmetric cryptography is needed, avoiding expensive public-key operations.

Figure~\ref{fig:anti-equivocation-latency} shows the median latency of the two non-equivocation mechanisms between a sender and two receivers.
The lowest latency is achieved by the fast path of our \tcb implementation, ranging from 2.2\,\us to 11\,\us depending on message size. 
In this path, \tcb leverages unanimous and timely participation of all processes, thereby avoiding signatures and preventing equivocation using just two rounds of \tbfull.
The slow path, triggered under failures, relies on public-key cryptography, which dominates latency and raises it to roughly 86\,\us.

Preventing equivocation using SGX requires accessing the enclave both at the sender and the receiver, in addition to broadcasting the message. 
Since our RDMA experimental setup lacks SGX, we approximate its cost using a machine with an Intel i7-7700K CPU (4.2\,GHz, 0.6\,GHz higher than our RDMA setup). 
Each enclave access takes 7--12.5\,\us, resulting in a minimum SGX-based latency of about 16\,\us, shown as the middle line in Figure~\ref{fig:anti-equivocation-latency}.

For both non-equivocation approaches, latency grows linearly with message size due to hashing and communication costs. 
\tcb's fast path is up to $6.5\times$ faster than the SGX solution, benefiting from ultra-low-latency communication.

\subsection{Impact of \tcb's Tail on Tail Latency}
\label{sec:tail-eval-latency}

We now study how the size of the tail in \tcb (parameter $t$) affects the client's tail latency.
We focus on \sysname's fast path and execute Flip with small 64\,B requests and larger 2\,KiB ones.
For both request sizes, we explore four \textit{tail} parameters.

Figure~\ref{fig:tail-memory-latency} shows the results.
For smaller values of $t$, we see a latency spike
indicative of thrashing as we move to higher
percentiles.
This spike occurs because
  \tcb uses a double-buffering mechanism with cryptographic
  \summaries (\S\ref{sec:smr-non-eq}) to switch between them; if both
  buffers fill before a \summary occurs
  (due to a small $t$), \tcb stalls.
The smaller the $t$, the sooner the buffers fill,
  the more often \tcb stalls, and
  hence the lower the percentile of
  the spike.
For small requests, a tail $t{=}128$ avoids thrashing up to the 99th percentile.
For larger requests, $t{=}64$ suffices, as filling the buffers takes more time, giving more time for the \summary to occur.

\begin{figure}
    \centering
    \includegraphics[width=\columnwidth]{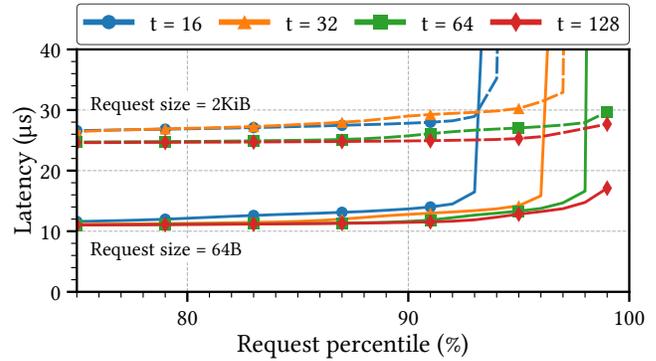}
    \caption{\sysname's tail latency for different \tcb tails $t$ for 2\,KiB requests (top) and 64\,B requests (bottom).}
    \label{fig:tail-memory-latency}
\end{figure}

\subsection{Memory Consumption}
\label{eval:mem}

Given the fundamental goal of \sysname to operate using finite memory,
we monitor the
consumption of disaggregated memory and
the local memory consumption at the leader replica, while re-running the experiment of Section~\ref{sec:tail-eval-latency}.

\begin{table}[b]
\caption{\sysname replica (top) and disaggregated (bottom) memory usage for different \tcb tails $t$ and request sizes.}
\vspace{3pt}
\label{table:memory}
\centering
\begin{tabular}{|c|c|c|c|c|}
\hline
\textbf{Request size} & \textbf{t = 16} & \textbf{t = 32} & \textbf{t = 64} & \textbf{t = 128} \\ \hline
64\,B  & 0.46\,GiB & 0.47\,GiB & 0.49\,GiB & 0.53\,GiB \\ \hline
2\,KiB & 4.3\,GiB  & 4.5\,GiB & 4.8\,GiB  & 5.5\,GiB  \\ \hline \hline
\textbf{Disag. Mem.} & 20\,KiB  & 40\,KiB & 81\,KiB  & 162\,KiB  \\ \hline
\end{tabular}
\end{table}

Table~\ref{table:memory} summarizes the results for two different request sizes (64\,B and 2\,KiB) and four different $tail$ parameters of \tcb (16, 32, 64, and 128).
For small 64\,B requests, the local memory consumption starts at 0.46\,GiB.
This is the entire memory that \sysname preallocates at startup and uses during its lifetime.
When \tcb's tail $t$ increases, \sysname's memory consumption increases linearly by $\approx$1\,MiB for each additional message in the tail.
For large 2\,KiB requests, the memory consumption starts at 4.3\,GiB ($t{=}16$) and increases at a rate of $\approx$11\,MiB per message.

\sysname consumes little disaggregated memory.
The last row of Table~\ref{table:memory} shows the memory used at a single memory node.
This amount is independent of the size of requests and depends only on \tcb's tail $t$:
  messages sent over \tcb are transmitted using our fast message-passing primitive;
  upon receiving a message, a receiver writes to disaggregated memory only the message's id and its fingerprint (a 32\,B cryptographic hash); the register
  implementation (\S\ref{sec:implementation-registers}) stores two copies, each with an 8\,B checksum.
To save space, registers use the identifiers of messages as their timestamps.

\section{Related Work} \label{sec:related}

\sysname uses RDMA to instantiate disaggregated memory. Prior work identifies some benefits and downsides of using RDMA in an adversarial environment. Aguilera \textit{et al.}~\cite{aguilera2019impact,aguileraBGPXZ21}
propose new RDMA-based techniques to enhance the resilience
and performance of BFT algorithms. That work is abstract
and far from practical solutions: it requires infinite memory and solves only single-shot consensus, stopping short of a solution for an SMR system. R\"{u}sch \textit{et al.}~\cite{rusch2018towards} design a Byzantine fault tolerant system that uses RDMA to improve performance, but requires $3f{+}1$ processes.

Prior work identifies vulnerabilities in the current
  generation of RDMA hardware and proposes ways to mitigate
  them~\cite{taranov2020srdma,rothenberger2021redmark}.
That work is orthogonal to \sysname and could be applied to it.
We expect that these problems will eventually be fixed with future NICs.

By separating execution from agreement~\cite{separating}, one can
  reduce the number of execution replicas to $2f{+}1$, but
  $3f{+}1$ replicas are still required for agreement.
With a black-box mechanism to prevent equivocation, only $2f{+}1$ replicas are required for BFT~\cite{correia2004half,clement2012limited,bendavid2022revisiting,ben2021brief}.
Several BFT systems achieve that using trusted hardware as the
  black box: attested append-only memory (A2M)~\cite{a2m} uses a trusted log, TrInc~\cite{levin2009trinc} and MinBFT~\cite{VeroneseCBLV13} use a trusted counter,
Hybster~\cite{behl2017hybrids} uses Intel's SGX,
  CheapBFT~\cite{KapitzaBCDKMSS12} uses FPGAs, and
H-MFT~\cite{zhuo2014mft} uses trusted hypervisors to implement
  write-once tables.
Avocado~\cite{bailleu2021avocado} implements a high-performance, replicated, confidential key-value store using CPU enclaves as the trusted computing base.

Blockchain systems also tolerate Byzantine
  failures, but their latency is in the seconds or minutes due to
  their heavy use of cryptography~\cite{algorand, byzcoin}, proof of work~\cite{nakamotobitcoin, bitcoin-ng}, and/or batching~\cite{tendermint, hotstuff}.
The recent SplitBFT~\cite{splitbft} uses SGX and $3f{+}1$ replicas to strengthen the safety and confidentiality of blockchains in public clouds.

While most of the prior work is not focused on
microsecond-scale latency (and hence came up with different
solutions from ours),
some recent SMR systems aim for lower latency.
Mu~\cite{aguilera2020microsecond} is highly optimized and provides
microsecond-scale performance, but tolerates only crash failures.
SBFT~\cite{gueta2019sbft} tolerates Byzantine failures and
  uses a fast path to improve latency,
  but does not achieve microsecond-scale performance due to its
  use of cryptography.
BFT SMR systems with $3f{+}1$ replicas can avoid cryptographic signatures, for example, in PBFT's optimized implementations~\cite{castro1999practical}. However, this is impossible in a system with $2f{+}1$ replicas~\cite{correia2004half}; the key to \sysname's performance is thus avoiding signatures on the fast path.

Carbink~\cite{carbink} and Hydra~\cite{hydra} build reliable disaggregated memory 
  to improve memory utilization in a cluster, albeit without support for
  concurrent shared access.
MIND~\cite{lee2021mind}, GAM~\cite{cai2018gam}, and Clover~\cite{tsai2020clover}, on the other hand, provide reliable shared memory, but they do not tolerate Byzantine writers.

\section{Discussion} \label{sec:discussion}

\paragraph{Does microsecond 2f+1 BFT require disaggregated memory?}
To achieve microsecond-scale BFT SMR, one must avoid the use of
  expensive signatures and trusted components in the critical path.
\sysname does so via a fast/slow path design that uses disaggregated
  memory in the slow path.
This leads to a small trusted computing base, but
  there may be other ways to achieve non-equivocation
  in \tcb's slow path without affecting fast-path 
  performance.

\paragraph{Can \sysname work with a Byzantine network?}
\sysname assumes network connections are authenticated and tamper-proof,
  which is realistic in data centers, where widely deployed protocols
  such as IPsec and SSL provide such guarantees at line rate.
What if such protocols are not available?
We can implement simple authenticated channels within \sysname
  without signatures in the critical path, by augmenting messages
  with an HMAC.
With BLAKE3, creating or verifying 256-bits 
   HMACs takes $\approx$100\,ns. 
As a result, we expect less than 2\,\us of additional overhead on the fast path of \sysname.

\paragraph{What about \sysname's throughput?}
Any system can provide a throughput that is inverse of its latency.
For \sysname, that amounts to ${\approx}91$\,kops for small 32\,B requests.
\sysname doubles this figure,
  by exploiting the slack between the processing of events in a consensus slot to interleave two requests with minimal latency penalty.
Throughput can be further optimized using well-known techniques,
  such as batching~\cite{danezis2022narwhal} and running parallel consensus instances on multiple cores~\cite{behl2017hybrids}, but we have not done so.

\section{Conclusion} \label{sec:conclusion}

\sysname is the first BFT SMR system to achieve
  microsecond-scale latency, $2f{+}1$ replicas, practically bounded
  memory, and a small non-tailored trusted computing base.
To do that, \sysname introduces \tcbfull 
  to prevent equivocation, and a matching consensus algorithm.
\sysname uses disaggregated memory as a trusted computing base,
  which our prototype implements using RDMA.
We applied \sysname to three applications to show that BFT can be realistic in data centers, even in latency-critical settings.
\section*{Acknowledgments} \label{sec:acknowledgments}
We thank our anonymous reviewers for their valuable comments,
as well as our anonymous artifact evaluators for reviewing our prototype implementation.
Finally, we thank our colleagues Jovan Komatovic and Pierre-Louis Roman for their feedback.
This work has been supported in part by SNSF Early Postdoc.Mobility grant P2ELP2\_195126.

\clearpage

\ifextended
\appendix
\section{Correctness of \tcbfull}\label{app:correct-TCB}

In this section, we provide a correctness argument for the implementation of \tcb given in Algorithm~\ref{alg:tcb}.

\begin{observation}\label{obs:one-per-k}
A correct broadcaster $p$ \tb-broadcasts at most one \texttt{LOCK} message and at most one \texttt{SIGNED} message per sequence number $k$. Moreover, both of these broadcasts hold the same message $m$.
\end{observation}

\begin{observation}\label{obs:one-locked-per-k}
A correct process $p$ \tb-broadcasts at most one \texttt{LOCKED} message per sequence number $k$.
\end{observation}
\begin{proof}
Correct processes only broadcast \texttt{LOCKED} messages at line~\ref{alg:tcb:locked}.
Moreover, \texttt{locks}$[k\%t]$, which is only modified at lines~\ref{alg:tcb:update-lock} and \ref{alg:tcb:update-lock-signed}, is updated strictly monotonically.
Thus, once the branch is entered at line~\ref{alg:tcb:check-higher-lock} (and thus \texttt{locks}$[k\%t]$ updated at line~\ref{alg:tcb:update-lock}), it cannot be entered for the same $k$, which ensures that line~\ref{alg:tcb:locked} is executed at most once per $k$ at correct processes. 
\end{proof}

\begin{lemma}[Tail-Validity] \label{lem:ctb:1}
If a correct process $p$ broadcasts $(k, m)$ and never broadcasts a message $(k', m')$ with $k' \ge k + t$, then all correct processes eventually deliver $(k, m)$.
\end{lemma}

\begin{proof}
Let $p,k,m$ be as in the statement of the lemma and let $q$ be a correct receiver. We will show that $q$ eventually delivers $(k, m)$, which is sufficient to prove the lemma.

Since $p$ is correct, $p$ \tb-broadcasts $\langle \texttt{SIGNED}, k, m, sig\rangle$ with a valid signature. Since both $p$ and $q$ are correct, $q$ eventually \tb-delivers it.

The $\langle \texttt{SIGNED}, k, m, sig\rangle$ message of $p$ will trigger at $q$ the event at line~\ref{alg:tcb:upon-signed}. 
Given that $p$'s signature is valid, the check at line~\ref{alg:tcb:valid-sig} succeeds.
By the premise, $p$ does not broadcast any message with sequence number $k' \ge k+t$, so \texttt{locks}$[k\%t]$, which is only modified at lines~\ref{alg:tcb:update-lock} and \ref{alg:tcb:update-lock-signed}, cannot contain a value greater than $k$. Moreover, since no other message $m$ is broadcast for $k$ (Observation~\ref{obs:one-per-k}), if \texttt{locks}$[k\%t]$ already contains $k$, it must also contain the message $m$. Thus, $q$ enters the if branch at line~\ref{alg:tcb:check-higher}. Since $p$ is correct and no process can forge $p$'s signature, no validly signed entry $(k,m')$ with $m\ne m'$ can exist in any process's SWMR slot, so $q$ does not return by triggering the check at line~\ref{alg:tcb:byzantine-check}. Finally, since $p$ does not broadcast any message with sequence number $k' \ge k+t$, no process's SWMR slot can contain a validly signed entry $(k',\cdot)$ with $k'>k$ and $k' \equiv k\ (\textrm{mod}\ t)$, so $q$ does not return by triggering the check at line~\ref{alg:tcb:tail-check}. Therefore, $q$ must call \texttt{deliver\_once}$(k,m)$ at line~\ref{alg:tcb:slow-delivery}. If this call does not deliver $(k,m)$, it must have been delivered before. Thus, $q$ eventually delivers $(k,m)$.
\end{proof}

\begin{lemma}[Agreement] \label{lem:ctb:2}
If $p$ and $q$ are correct processes, $p$ delivers $(k, m)$ from $r$, and $q$ delivers $(k, m')$ from $r$, then $m=m'$.
\end{lemma}

\begin{proof}
Assume towards a contradiction that $m\ne m'$. We consider two cases: (1) at least one process delivers via the fast path, and (2) both processes deliver via the slow path.

In case (1), assume \textit{wlog} that $p$ delivers via the fast path. Then $p$ must have \tb-delivered a \texttt{LOCKED} message from $q$ for $(k,m)$. So $q$ must have \tb-broadcast a \texttt{LOCKED} message at line~\ref{alg:tcb:locked}. By Observation~\ref{obs:one-locked-per-k}, $q$ cannot have broadcast a \texttt{LOCKED} message for $(k,m')$.
Thus, it cannot have delivered $m'$ via the fast path.
Moreover, $q$ must have put $m$ in \texttt{locks}$[k\%t]$ at line~\ref{alg:tcb:update-lock}. Thus $q$ cannot enter the if branch at line~\ref{alg:tcb:check-higher} and cannot deliver $(k, m')$ via the slow path either, hence a contradiction.

In case (2), assume \textit{wlog} that $p$ writes $(k,sig,m)$ to \texttt{SWMR}$[p][k\%t]$ (line~\ref{alg:tcb:copy}) before $q$ writes $(k,sig',m')$ to \texttt{SWMR}$[q][k\%t]$. Thus, when $q$ reads $p$'s $k\%t$ slot at line~\ref{alg:tcb:for}, $q$ sees either (i) $(k,sig,m)$ or (ii) $(k'',\cdot,\cdot)$ with $k'' > k$, and $k'' \equiv k\ (\textrm{mod}\ t)$. In case (i), $q$ will return by triggering the check at line~\ref{alg:tcb:byzantine-check}, and thus not deliver $(k,m')$, a contradiction. In case (ii), $q$ will return by triggering the check at line~\ref{alg:tcb:tail-check}, and thus not deliver, a contradiction.
\end{proof}

\begin{lemma}[Integrity] \label{lem:ctb:3}
If a correct process delivers $(k, m)$ from $p$ and $p$ is correct, $p$ must have broadcast $(k, m)$.
\end{lemma}

\begin{proof}
Let $p$ and $q$ be correct processes and assume $q$ delivers $(k, m)$ from $p$. There are two possible cases: (1) $q$ delivers using the fast path at line~\ref{alg:tcb:fast-delivery}, or (2) $q$ delivers using the slow path at line~\ref{alg:tcb:slow-delivery}. 

In case (1), $q$ must have \tb-delivered \texttt{LOCKED} messages for $(k,m)$ from all processes, including itself. Therefore $q$ must have  \tb-broadcast a \texttt{LOCKED} message for $(k,m)$ at line~\ref{alg:tcb:locked} after \tb-delivering a \texttt{LOCK} message for $(k,m)$ from $p$. Thus, $p$ must have \tb-broadcast a \texttt{LOCK} message, which $p$ can only do as part of the \tcb-broadcast call. So $p$ must have broadcast $(k,m)$.

In case (2), $q$ must have \tb-delivered a valid \texttt{SIGNED} message from $p$ for $(k,m)$. Since $p$ is correct and no process can forge its signature, $p$ must have broadcast a \texttt{SIGNED} message for $(k,m)$. So $p$ must have broadcast $(k,m)$.
\end{proof}

\begin{lemma}[No duplication] \label{lem:ctb:4}
No correct process delivers $(k, *)$ from $p$ twice.
\end{lemma}

\begin{proof}
Correct processes only deliver through \verb|deliver_once|. Lines~\ref{alg:tcb:if-delivered} and \ref{alg:tcb:update-delivered} ensure that a correct process only triggers $deliver$ at most once per sequence number $k$.
\end{proof}

\begin{theorem}
From Lemmas~\ref{lem:ctb:1}, \ref{lem:ctb:2}, \ref{lem:ctb:3} and \ref{lem:ctb:4}, Algorithm~\ref{alg:tcb} implements all the properties of \tcbfull. \qed
\end{theorem}
\section{Correctness of Consensus}\label{app:correct-consensus}

\newcommand{\slot}{s}
\newcommand{\req}{r}
\newcommand{\inview}{view}

\begin{lstlisting}[float=ht!,caption={Common Case (stable leader)},label={alg:consensus}]
@\textit{CTBcasts FIFO-deliver and block upon a Byzantine message}@

upon Init:
  @\inview@ = 0
  next_slot = 0
  checkpoint = (app_state: Initial, open_slots: [0, 99])@$_\Sigma$@
  for each replica:
    state[replica] = {
      view = 0, seal_view = @$\bot$@, new_view = @$\bot$@,
      prepares: Map<slot, PREPARE> = {},
      commits: Map<slot, COMMIT> = {},
      checkpoint = (Initial, [0, 99])@$_\Sigma$@ }

def Propose(req):
  wait (leader(@\inview@) == me and next_slot in checkpoint.open_slots and NEW_VIEW broadcast i@\phantom{}@f view @$>$@ 0)
  CTBcast-bcast <PREPARE, @\inview@, next_slot++, req>

upon CTBcast-dlvr <PREPARE, v, @\slot@, @\req@> from p as P:@\label{alg:cons:deliver-prepare}@
  state[p].prepares[@\slot@] = P
  if v != @\inview@ or @\slot@ @$\notin$@ checkpoint.open_slots: return
  TBcast-bcast <WILL_CERTIFY, v, @\slot@> # Fast path
  TBcast-bcast <CERTIFY, sign(P)> # Slow path

# Fast path
upon TBcast-dlvr <WILL_CERTIFY, v, @\slot@> from 2f+1:
  if v != @\inview@ or @\slot@ @$\notin$@ checkpoint.open_slots: return
  TBcast-bcast <WILL_COMMIT, v, @\slot@>

upon TBcast-dlvr <WILL_COMMIT, v, @\slot@> from 2f+1:
  if v != @\inview@ or @\slot@ @$\notin$@ checkpoint.open_slots: return
  trigger once Decide(@\slot@, state[leader(v)].prepares[s].req)@\label{alg:cons:fast-decide}@
  
# Slow path
upon TBcast-dlvr <CERTIFY, <P., v, @\slot@, _>@$_\sigma$@> from f+1 as P@$_\Sigma$@:
  if v != @\inview@ or @\slot@ @$\notin$@ checkpoint.open_slots: return
  CTBcast-bcast <COMMIT, P@$_\Sigma$@>
  
upon CTBcast-dlvr <COMMIT, P@$_\Sigma$@> from p as C:
  state[p].commits[P@$_\Sigma$@.slot] = C
  if dlvred f+1 COMMIT with a matching PREPARE:
    trigger once Decide(P@$_\Sigma$@.slot, P@$_\Sigma$@.req)@\label{alg:cons:slow-decide}@

# Checkpoints
after having decided on all checkpoint.open_slots:@\label{alg:cons:every-x}@
  wait fo@\phantom{}@r all decided requests to be applied to the App
  next_cp = (App.Snapshot(), checkpoint.open_slots + 100)
  TBcast-bcast <CERTIFY_CHECKPOINT, sign(next_cp)>

upon TBcast-dlvr <CERTIFY_CHECKPOINT, c@$_\sigma$@> from f+1 as C@$_\Sigma$@:
  MaybeCheckpoint(C@$_\Sigma$@)

upon CTBcast-dlvr <CHECKPOINT, C@$_\Sigma$@> from p as CP:
  state[p].checkpoint = CP
  forget state[p].commits an@\phantom{}@d prepares @$\notin$@ C@$_\Sigma$@.open_slots
  MaybeCheckpoint(C@$_\Sigma$@)
  
def MaybeCheckpoint(C@$_\Sigma$@):
  if C@$_\Sigma$@ supersedes checkpoint:
    checkpoint = C@$_\Sigma$@
    App.BringUpToSpeed(checkpoint)
    TBcast-bcast <CHECKPOINT, checkpoint>
\end{lstlisting}

\begin{lstlisting}[float=ht!,caption={View Change},label={alg:consensus-view-change}]
upon suspicion of leader(view): ChangeView()

def ChangeView():
  for each <WILL_COMMIT, v@$|_{v==view}$@, @\slot@> bcast via TBcast:
    wait to have broadcast a matching COMMIT o@\phantom{}@r CHECKPOINT
  CTBcast-bcast <SEAL_VIEW, @++\inview@>
  
upon CTBcast-dlvr <SEAL_VIEW, v> from p as SV:
  state[p].seal_view = SV
  state[p].view = v
  send@$_{leader(v)}$@ <CRTFY_VC, v, sign((p, state[p]\new_view))>

upon dlvr f+1 matching <CRTFY_VC, v@$|_{v==view}$@, s@$_\sigma$@> about f+1 replicas as C:
  if me != leader(view): return
  CTBcast-bcast <NEW_VIEW, C>
  MaybeCheckpoint(highest checkpoint in C)
  for s in checkpoint.open_slots:
    CTB-bcast <PREPARE, v, s, MustPropose(s, C)>@\label{alg:cons:vc-prepare}@
  next_slot = checkpoint.open_slots.last + 1

upon CTBcast-dlvr <NEW_VIEW, certificates> from p as NV:
  state[p].new_view = NV
  while @\inview@ != NV.view + 1: ChangeView()
    
def MustPropose(slot, certificates):
  if slot > max open slot i@\phantom{}@n certificates: return Any
  return latest committed req fo@\phantom{}@r slot i@\phantom{}@n certificates or @$\bot$@
\end{lstlisting}

\begin{lstlisting}[float=ht!,caption={\tcb Summaries},label={alg:consensus-tcb-summaries}]
after CTBcast-dlvr the message with id % tail == 0 from p:
  send@$_p$@ <CERTIFY_SUMMARY, sign((p, id, state[p]))>

every tail invocations of CTBcast-bcast:@\label{alg:cons:summary-every}@
  block calls to CTBcast-bcast

upon dlvr <CERTIFY_SUMMARY, (me, id, _)@$_\sigma$@> from f+1 as S@$_\Sigma$@:
  TBcast-bcast <SUMMARY, S@$_\Sigma$@>
  unblock calls to CTBcast-bcast@\label{alg:cons:summary-unblock}@

upon TBcast-dlvr <SUMMARY, (p, id, history)@$_\Sigma$@>:
  when a gap is detected i@\phantom{}@n the dlvry of CTBcast from p:
    if the latest message dlvred fro@\phantom{}@m p is lower than id:
      dlvr i@\phantom{}@n order p's missed CTBcast messages i@\phantom{}@n history without running the Byzantine checks @(Alg.~\ref{alg:consensus-byz-checks})@
      continu@\phantom{}@e dlvring p's CTBcast messages afte@\phantom{}@r id
\end{lstlisting}

\begin{lstlisting}[float=ht!,caption={\tcb's Byzantine Checks},label={alg:consensus-byz-checks}]
def valid <PREPARE, v, @\slot@, @\req@> from p:
  state[p].view == v and leader(v) == p and
    s in state[p].checkpoint.open_slots and
    p never prepared slot @\slot@ before i@\phantom{}@n v and
    (v == 0 or (state[p].new_view != @$\bot$@ and
      @\req@ == MustPropose(s, state[p].new_view)))

def valid <COMMIT, P@$_\Sigma$@> from p as C:
  P@$_\Sigma$@.slot in state[p].checkpoint.open_slots and
    P@$_\Sigma$@.view == state[p].view and
    state[p].commits[P@$_\Sigma$@.slot] != C

def valid <CHECKPOINT, C@$_\Sigma$@> from p:
  C@$_\Sigma$@ supersedes state[p].checkpoint
  
def valid <SEAL_VIEW, v> from p:
  state[p].view < v
  
def valid <NEW_VIEW, certificates> from p:
  leader(state[p].view) == p and
    it is p's first non-CHECKPOINT message i@\phantom{}@n this view and
    each certificate is about a different replica and
    each certificate is s@$\phantom{}$@igned by f+1 different replicas and
    each certificate is about view state[p].view
\end{lstlisting}
This section gives the pseudocode of \sysname's consensus alongside a correctness argument.
Algorithm~\ref{alg:consensus} describes \sysname's operation under a stable leader. 
Algorithm~\ref{alg:consensus-view-change} describes view changes.
Algorithm~\ref{alg:consensus-tcb-summaries} describes how \summaries let \sysname handle the gaps
caused by \tcb's tail-validity.
Finally, Algorithm~\ref{alg:consensus-byz-checks} gives the explicit requirements for messages to pass \tcb's Byzantine checks.

\subsection{Validity}
\begin{lemma}\label{lem:validity}
For a fixed slot $s$, with no faulty processes, if some process $p$ delivers and accepts $\langle \texttt{PREPARE}, v, s, r\rangle$ in Algorithm~\ref{alg:consensus} at line~\ref{alg:cons:deliver-prepare}, then $r$ must have been proposed by some correct process.
\end{lemma}
\begin{proof}[Proof Sketch]
We will prove the lemma by induction on the view $v$ in which $p$ accepts the \texttt{PREPARE} message.
The base case is $v=0$. Process $p$ must have delivered a \texttt{PREPARE} message for $r$ from the leader $\ell_0$ of view $0$. Since $\ell_0$ is correct, it only sends \texttt{PREPARE} messages for values that are in its input, or for values that are part of a valid view change certificate from the previous view. Since there is no previous view in view $0$, it must be that $r$ was $\ell_0$'s input.

Now, for the induction step, assume that the lemma is true up to view $v$, and examine the case in which $p$ accepts $\langle\texttt{PREPARE}, v+1, s, r\rangle$ in view $v+1$. All processes are assumed to be correct, so the \texttt{PREPARE} message must have been sent by $\ell_{v+1}$, the leader of view $v+1$. Correct processes only send one \texttt{PREPARE} message per slot per view, so $\ell_{v+1}$ must have sent $\langle\texttt{PREPARE}, v+1, s, r\rangle$ either as a new proposal, or during the view change from $v$ to $v+1$, in Algorithm~\ref{alg:consensus-view-change} at line~\ref{alg:cons:vc-prepare}.
In the first case, $r$ is by definition proposed by a correct process as part of $\ell_{v+1}$'s input.
In the second case, $r$ must be a valid value (i.e., be returned by \texttt{MustPropose}), given the view change certificates for view $v$. There are two cases in which $r$ is such a valid value: (1) one of the certificates contains a \texttt{COMMIT} messages for $r$ in $v'$ with $ v' \le$ $v$, or (2) none of the certificates contain a \texttt{COMMIT} message for $r$, and $r$ is the input of $\ell_{v+1}$. In case (1), a quorum of processes must have delivered and accepted $\langle\texttt{PREPARE}, v', s, r\rangle$ messages in view $v'$ with $v' \le v$ and thus, by induction, $r$ must have been proposed by some correct process. In case (2), $r$ is also proposed by a correct process. This concludes the induction step and the proof.
\end{proof}

\begin{theorem}[Weak Validity]
For a fixed slot $s$, with no faulty processes, if some process $p$ decides value $r$ in $s$, then $r$ must have been proposed by some correct process.
\end{theorem}
\begin{proof}[Proof Sketch]
Process $p$ may decide $r$ either at (1) line~\ref{alg:cons:fast-decide} (fast path), or (2) at line~\ref{alg:cons:slow-decide} (slow path) of Algorithm~\ref{alg:consensus}. Let $v$ be the view in which $p$ decides $r$. In case (1), $p$ must have delivered and accepted a $\langle\texttt{PREPARE}, v, s, r\rangle$ in view $v$, so by Lemma~\ref{lem:validity}, $r$ must have been proposed by some correct process. In case (2), $p$ must have received valid \texttt{COMMIT} messages for $r$ from a quorum. Thus, a quorum of processes must have delivered and accepted a \texttt{PREPARE} message for $r$, so by Lemma~\ref{lem:validity}, $r$ must have been proposed by some correct process.
\end{proof}

\subsection{Agreement}

\begin{observation}\label{obs:conflicting-prepares}
For a fixed slot $s$ and view $v$, two correct processes never deliver and accept conflicting \texttt{PREPARE} messages.
\end{observation}
\begin{proof}[Proof Sketch]
Correct processes deliver and accept \texttt{PREPARE} messages only when coming from the leader. Furthermore, they deliver and accept at most one \texttt{PREPARE} message per view per slot. Thus, by the Agreement property of \tcb, if two correct processes deliver and accept \texttt{PREPARE} messages in the same view, then those messages are for the same value and thus do not conflict.
\end{proof}

\begin{corollary}\label{cor:no-conflicting-commits}
For a fixed slot $s$ and view $v$, two processes never broadcast conflicting valid \texttt{COMMIT} messages. 
\end{corollary}
\begin{proof}[Proof Sketch]
Assume towards a contradiction that two processes $q$ and $p$ broadcast conflicting valid \texttt{COMMIT} messages.
Given that each valid \texttt{COMMIT} message is made of a quorum of valid \texttt{CERTIFY} messages, $p$ and $q$ must have delivered two quorums of valid \texttt{CERTIFY} messages about different \texttt{PREPARE} messages.
By definition, each quorum must contain one correct process.
Moreover, a correct process only broadcasts a \texttt{CERTIFY} message about the \texttt{PREPARE} message it delivered.
Thus, two correct processes delivered different \texttt{PREPARE} messages for the same slot in the same view.
This contradicts Observation~\ref{obs:conflicting-prepares}.
\end{proof}

\begin{corollary}
For a fixed slot $s$, 
the view change certificates corresponding to two processes cannot have conflicting \texttt{COMMIT} messages from the same view.
\end{corollary}
\begin{proof}[Proof Sketch]
Assume not.
Then there exist processes $p_1$ and $p_2$ such that their view change certificates at the end of view $v$ are conflicting: they contain different \texttt{COMMIT} messages for values $r_1$ and $r_2$, respectively, from the same view.
Since each certificate contains an approval from a quorum, each certificate must have been approved by at least one correct process.
Thus, a correct process must have received a \texttt{COMMIT} from $p_1$ for $r_1$ and, in the same view, a correct process (not necessarily the same) must have received a \texttt{COMMIT} from $p_2$ for $r_2$. By the Integrity property of \tcb, this implies that $p_1$ and $p_2$ must have sent conflicting commits for the same slot and view, which is impossible by Corollary~\ref{cor:no-conflicting-commits}.
\end{proof}

\begin{lemma}\label{lem:no-subsequent-prepare}
For a fixed slot $s$ and view $v$, if a quorum broadcasts \texttt{COMMIT} messages for the same value $r$, then no correct process accepts a \texttt{PREPARE} message for any other value $r'\ne r$ in any view $v'\ge v$.
\end{lemma}
\begin{proof}[Proof Sketch]
We proceed by induction on $v'$. The base case is $v'=v$. Since a quorum broadcasts \texttt{COMMIT} messages for $r$ in view $v$, at least one correct process $p$ must have broadcast a \texttt{COMMIT} for $r$. Thus, some correct process must have delivered and accepted a \texttt{PREPARE} message for $r$. Thus, by Observation~\ref{obs:conflicting-prepares}, no correct process may accept a \texttt{PREPARE} for a different value $r'\ne r$ in the same view.

Now, for the induction step, assume the lemma is true up to view $v'$, and assume that in view $v'+1$, some correct process $p$ accepts a \texttt{PREPARE} message for some other value $r'\ne r$. For this to happen, $r'$ must be a valid value according to the view change certificates provided by the leader $\ell_{v'+1}$ of view $v'+1$. Thus, at least one process $q$ must have sent a \texttt{COMMIT} message $C$ for $r'$ in a view $v''\leq v'$. Furthermore, $C$ must have been accepted by at least one correct process $w$, in order for $q$'s state to have been certified by a quorum. In order for $w$ to accept $C$, $C$'s corresponding \texttt{PREPARE} message must have been certified, and thus accepted, by at least one correct process in view $v''$. This contradicts our induction hypothesis. So it is impossible for any correct process to accept a \texttt{PREPARE} message for $r'$ in view $v'+1$. This completes the induction step and the proof. 
\end{proof}

\begin{theorem}[Agreement]
For a given slot $s$, correct processes cannot decide different values.
\end{theorem}

\begin{proof}[Proof Sketch]
Assume by contradiction that there exist two correct processes $p_1$ and $p_2$, such that $p_1$ decides $r_1$ in view $v_1$ and $p_2$ decides $r_2\ne r_1$ in view $v_2$. Assume further \textit{wlog} that $v_1 \leq v_2$. We consider four cases, based on whether $p_1$ and $p_2$ decide on the fast path or the slow path.

\paragraph{Case 1: Fast-fast.} Both $p_1$ and $p_2$ decide their respective values on the fast path. If $v_1 = v_2$, then $p_1$ and $p_2$ must have accepted conflicting \texttt{PREPARE}s in the same view, which is impossible by Observation~\ref{obs:conflicting-prepares}. Otherwise, if $v_1 < v_2$, then at least $f{+}1$ correct processes (a quorum) must have broadcast \texttt{COMMIT} messages for $r_1$ before sealing view $v_1$. Thus, by Lemma~\ref{lem:no-subsequent-prepare}, no correct process can accept a \texttt{PREPARE} for $r_2$ in $v_2$, so $p_2$ cannot decide $r_2$ on the fast path in $v_2$.

\paragraph{Case 2: Fast-slow.} $p_1$ decides on the fast path and $p_2$ decides on the slow path. Then, $p_1$ must have accepted a \texttt{PREPARE} for $r_1$ in view $v_1$ (call this Fact~1). Moreover, $p_2$ must have accepted \texttt{COMMIT} messages for $r_2$ from a quorum. This implies that a quorum broadcast \texttt{COMMIT} messages for $r_2$ in some view $v_2'\leq v_2$ (call this Fact 2). If $v_2' \leq v_1$, then we reach a contradiction with Fact 1 by Lemma~\ref{lem:no-subsequent-prepare}. If $v_2' > v_1$, then a quorum of correct processes must have broadcast \texttt{COMMIT} messages for $r_1$ before sealing $v_1$; thus, by Lemma~\ref{lem:no-subsequent-prepare}, we reach a contradiction with Fact 2, since no correct process could have accepted a \texttt{PREPARE} for $r_2$ in $v_2$. 

\paragraph{Case 3: Slow-fast.} This case is symmetric with Case 2 above.

\paragraph{Case 4: Slow-slow.} If both $p_1$ and $p_2$ decide on the slow path, then both processes must have accepted \texttt{COMMIT} messages from a quorum. Let $v_1'$ and $v_2'$ be the views in which the \texttt{COMMIT} messages accepted by $p_1$ and $p_2$, respectively, were sent. Assume \textit{wlog} that $v_1' \leq v_2'$. Then, by Lemma~\ref{lem:no-subsequent-prepare}, no correct process could have accepted a \texttt{PREPARE} for $r_2$ in view $v_2$. Thus, no correct process could have sent a \texttt{COMMIT} message for $r_2$ in $v_2$, and thus it is impossible for a quorum to have sent \texttt{COMMIT} messages for for $r_2$ in $v_2$.

We have reached a contradiction in all four cases. This completes the proof of the theorem.
\end{proof}

\subsection{Liveness}

In this section, we provide informal arguments for the liveness of our protocol. We assume that the system is eventually synchronous and that correct processes propose values infinitely often. 
Intuitively, liveness is ensured by three mechanisms: (1) the view change mechanism, (2) checkpoints and (3) \tcb \summaries.
First, we assume that \tcb messages are not dropped and explain why liveness is ensured by the first two mechanisms.
Then, we complete our explanation with the way \tcb summaries help overcome the problem of dropped messages.

The view change mechanism ensures that at least one correct process is able to decide forever.
Assume towards a contradiction that all correct processes stop deciding.
Then, as long as they do not make progress, correct processes will change view thanks to the view change protocol in Algorithm~\ref{alg:consensus-view-change}.
Eventually, after the global stabilization time (GST), all correct processes are guaranteed to (1) reach a view $v$ in which the leader is correct, and (2) communicate with each other in a timely manner.
Thus, given that the timely collaboration of $f{+}1$ processes is enough for the common path described in Algorithm~\ref{alg:consensus} to be live, the correct replicas decide, hence a contradiction.

However, having a single correct process deciding infinitely often is not enough for the overall system to make progress as clients need to obtain a response from $f{+}1$ processes.
The checkpoint mechanism guarantees that, if a correct process $p$ decides on slots infinitely often, then all correct processes also make progress.
This is because, in order to keep on making progress (and thus maintain correct processes under its control), the leader of a view is mandated to broadcast a \texttt{CHECKPOINT} message periodically.
This message is then re-broadcast by the potentially single correct process in the view and, after GST, delivered by all correct processes.
Because checkpoints are transferable, when another correct process receives a checkpoint, it is able to decide on the slots contained in the checkpoint and bring its application state up to speed with the latest decided requests.

Lastly, \tcb \summaries ensure that, if the system were to be reduced to only $f{+}1$ correct processes, they would be able to continue making progress in spite of \tcb's delivery gaps.
The worst scenario is arguably the one in which a correct process $p$ used to make progress with $f$ faulty processes before being let down by them, and the other $f$ correct processes ending up with a gap in their \tcb delivery of $p$'s messages due to asynchrony.
In this case, %
Algorithm~\ref{alg:consensus-tcb-summaries} ensures that $p$ will not risk creating a gap %
before having obtained a \summary to help overcoming it.
Using this \summary%
, $p$ can let correct processes continue delivering its messages by convincing them that they will not violate the safety of the protocol due to missed messages.
Moreover, $p$ will always obtain a new \summary: either it will be helped by Byzantine processes, or, after GST, it will get help from correct processes by combining the previous \summary with the last $t$ messages it broadcast.
\fi

\bibliographystyle{ACM-Reference-Format}
\bibliography{references}

\end{document}